\documentclass[11pt,a4paper]{article}
\usepackage{amsmath,amssymb,amsfonts,amsthm}
\usepackage{fullpage}
\usepackage{paralist}
\usepackage{dsfont} 
\usepackage{graphicx}
\usepackage{mathabx} 
\usepackage{booktabs}
\usepackage{bm}
\usepackage{bbm}
\usepackage[ansinew]{inputenc}
\usepackage[colorlinks=true,linkcolor=blue,citecolor=blue,pdfborder={0 0 0}]{hyperref}
\usepackage{color}
\usepackage{float}

\usepackage{natbib}

\newtheorem{theorem}{Theorem}
\newtheorem{lemma}[theorem]{Lemma}
\newtheorem{proposition}[theorem]{Proposition}

\begin{document}

\title{Dealing with seasonal variability and inter-site dependence in regional flood frequency analysis}


\author{Paul Kinsvater\,\footnote{Corresponding author. Technische Universit\"at Dortmund, Fakult\"at Statistik. {E-mail:} \texttt{kinsvater@statistik.tu-dortmund.de}.}
,~
Friederike Deiters\,
~and 
Roland Fried
}
\date{}

\maketitle

\section*{Abstract}

This paper considers the regional estimation of high quantiles of annual maximal river flow distributions $F$, an important problem from flood frequency analysis. Even though this particular problem has been addressed by many papers, less attention has been payed to incorporating seasonal variability and spatial dependence into the methods. We are going to discuss two regional estimators of high quantiles of local distributions $F$ that are able to deal with these important features, namely, a parametric approach based on so-called two-component extreme value distributions and a semi-parametric approach based on regional estimation of a tail index. The asymptotic normality of the estimators is derived for both procedures, which for instance enables us to account for estimation uncertainty without the need of parametric dependence models or bootstrap procedures. A comprehensive simulation study is conducted and our main findings are illustrated on river flow series from the Mulde basin in Germany, where people have suffered several times from severe floods over the last 100 years.

\section{Introduction}\label{sec:intro}

Flood frequency analysis (FFA) deals with the estimation of river flow distributions. A flow value is the amount of water (in cubic meter per second, m$^3$/s) passing a measurement station. The ultimate objective is to determine the design of future flood protection systems, for example, the height of a dam for some predefined non-failure probability. Common rules found in many official guidelines \citep[e.g.,][]{dwa_merkblatt_2012} focus on the distribution $F(x)=\mathbb{P}(X\leq x)$ of annual maximal flows $X$ at the site of interest and determine the height of a dam according to a high quantile $q=F^{-1}(p)$, for a given $p\in(0,1)$. Often experts have to deal with rather high probabilities $p\geq0.99$, depending on the safety-relevance of the local environment. \\
In practice the distribution $F$ can be estimated from annual maximal flows $X_1,\ldots,X_n$ of the past $n$ years. A serious problem in FFA is that only small sample lengths $n$ are available at each site. For a potential reduction of the estimation uncertainty, so-called regional methods are applied. These procedures combine observations from a set of sites $j\in\{1,2,\ldots,d\}$ sharing similar site characteristics related to the flood magnitude. For instance, the Index Flood method proposed in the seminal work by \citet{Dal60} is based on the hypothesis that the quantile functions are all identical up to an unknown local scale $s_j=s(F_j)>0$,
\begin{align}\label{eq:IFassumption}
\mathcal{H}_{0,\text{IF}}:\ F_j^{-1}=s_j\cdot G_{\theta}^{-1}\quad \forall  j=1,\ldots,d,
\end{align}
with distribution $G_\theta$ known up to a finite-dimensional parameter $\theta$. Thanks to the extreme value theorem \citep{FisTip1928}, a recommended choice for the parametric model $G_\theta$ is a generalized extreme value (GEV) distribution function
\begin{align}\label{eq:GEV}
	G_\theta(x)=\exp \left[ -\left(1+\xi\ \frac{x-\mu}{\sigma}\right)^{-1/\xi} \right] \text{ for } 1+\xi \ \frac{x-\mu}{\sigma}>0,
\end{align}

\begin{figure}[b!]
	\centering
	\includegraphics[width=.85\linewidth]{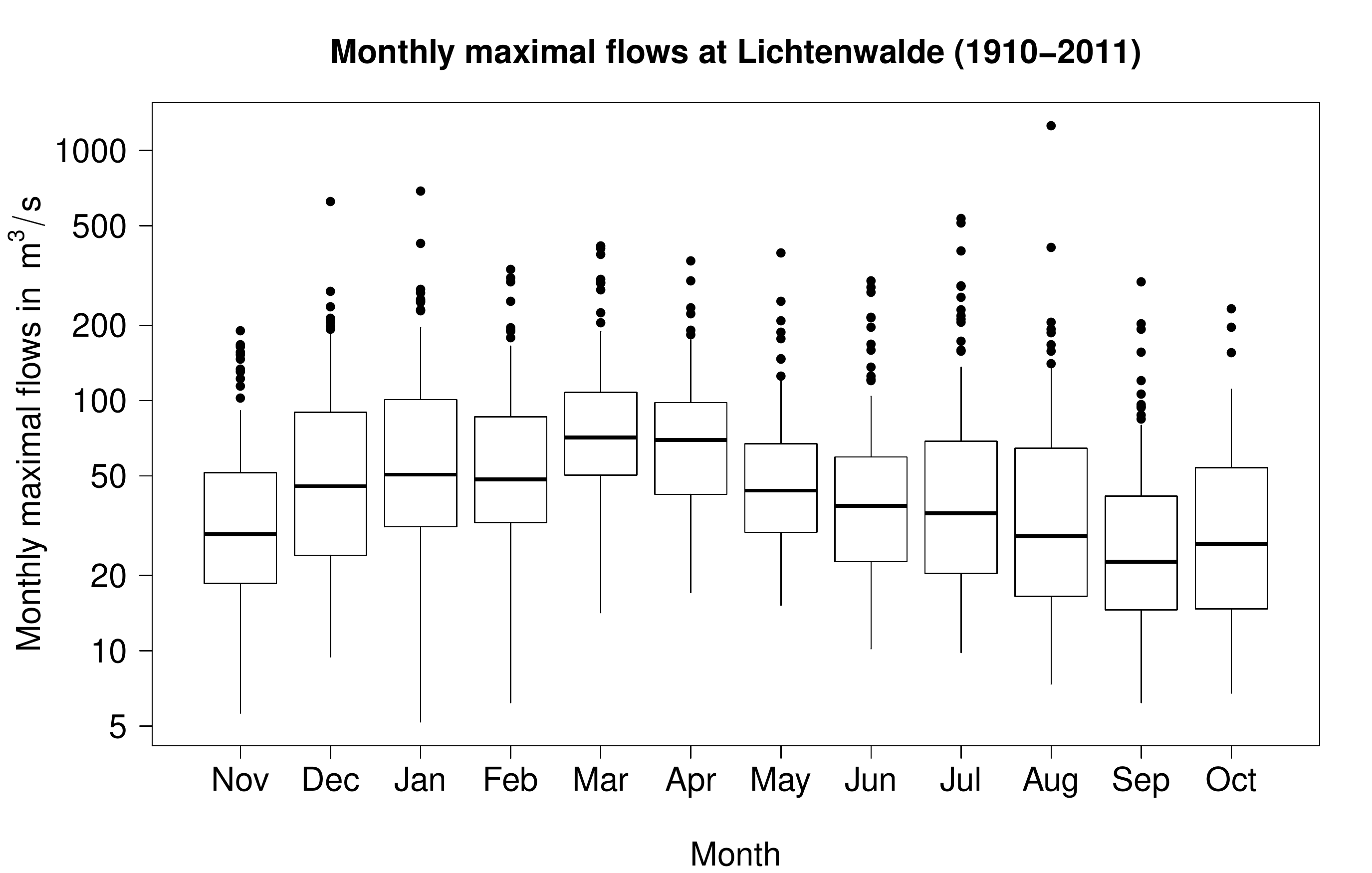}
	\caption{Boxplots of monthly maximal flows at station Lichtenwalde in Saxony, Germany, during the hydrological years 1910--2011. Monthly maximal flows are depicted on the $y$-axis on $\log$-scale.}
	\label{fig:monHQ}
\end{figure}

with parameters $\theta=(\mu,\sigma,\xi)'\in\mathbb{R}\times\mathbb{R}_+\times\mathbb{R}$ called location, scale and shape. These distributions arise as the only possible non-degenerate limit laws of standardized block maxima $X^{(b)}=\max\{Z_1,\ldots,Z_b\}$ over independent and identically distributed variables $Z_i$ as the block size $b$ tends to infinity. \\
However, it is also important to account for inter-site dependence. Besides the fact that stronger positive dependence between stations reduces the accuracy of regional estimators \citep{Ste83}, we need to estimate the correlation between local statistics to be able to consistently estimate the precision of regional approaches \citep{MarSte02,CunBur06}. Contrary to the previous references, we do not apply Monte Carlo or Bootstrap procedures. Instead we use direct estimates of correlation based on recent theory developed by \citet{LilKinFri16}.

\smallskip

Treating an annual maximal flow $X=\max\{Z_{\text{Jan}},Z_{\text{Feb}},\ldots,Z_{\text{Dec}}\}$ as a maximum over twelve independent and identically distributed monthly maximal flows $Z_{\text{month}}$ motivates the assumption that $F$ equals $G_\theta$ approximately for some unknown $\theta=(\mu,\sigma,\xi)'$ due to the finite block length of $b=12$. We even might assume that $F=G_\theta$ holds, which allows to estimate $F$ by parametric methods. A theoretical justification of this simplification is studied by \citet{FerDeh15} for estimation based on $L$-moments and by \citet{BueSeg15} for maximum likelihood estimation, where the asymptotic normality of the methods is derived for block maxima instead of exact GEV distributions (with block size tending to infinity). \\
The former simplification of identically distributed monthly maximal flows $Z_{\text{Jan}},\ldots,Z_{\text{Dec}}$ is not realistic. River flows, similar to temperature and rainfall, are subject to seasonal variability. Flows in the winter/spring season are fed by large masses of melting snow, while floods in the summer/autumn period usually are caused by short but heavy rainfalls. For empirical evidence of the seasonal variability we refer to Figure \ref{fig:monHQ}. The boxplots illustrate the distribution of monthly maximal flows for each specific month at station Lichtenwalde. The plot is computed from observations during the years 1910--2011. A seasonal pattern in location (median) and dispersion (interquartile range) of the boxes becomes apparent. This confirms that different physical mechanisms (melting snow, heavy rainfall) cause different river flow distributions. Summing up, it is not clear whether the assumption $F\approx G_\theta$ is justified, at least approximately, for annual maximal flow distributions. \\
An approach that is particularly designed to account for the different flood generating mechanisms proceeds as follows: Suppose that our observations include seasonal maxima $W$ and $S$ from two disjoint seasons, e.g. hydrological winter and summer, such that $X=\max\{W,S\}$. Both, $W$ and $S$, represent maximal values that are computed from more homogeneous periods of, say, six monthly maximal observations with distribution functions $\mathbb{P}(W\leq x)=F_w(x)$ and $\mathbb{P}(S\leq x)=F_s(x)$. We thus may assume that $F_w\approx G_{\theta_w}$ and $F_s\approx G_{\theta_s}$ for some parameter vectors $\theta_w$ and $\theta_s$. Furthermore, assuming independence between $W$ and $S$, we obtain
\begin{align}\label{eq:productdist}
F(x)=\mathbb{P}(\max\{W,S\}\leq x)=\mathbb{P}(W\leq x~,S\leq x)=F_w(x)\cdot F_s(x)\approx G_{\theta_w}(x)\cdot G_{\theta_s}(x)
\end{align}
for the distribution function $F$ of annual maximal flows. Again, we may assume $F(x)=G_{\theta_w}(x)\cdot G_{\theta_s}(x)$ for practical reasons. A distribution function $F=G_{\theta_w}\cdot G_{\theta_s}$ is called two-component GEV \citep{RosFioVer84,GabArn91,StrKoc12,RulBuiRotKys16}. \\
In the mentioned references it is stated that such a seasonal approach potentially leads to more plausible estimates of high quantiles of $F$, where plausibility may be interpreted as a smaller number of observations being classified as outliers \citep{RosFioVer84}. On the other hand, seasonal GEV estimation means that maxima over blocks with a smaller size are considered, say, $b=6$ instead of $b=12$. Is it still plausible to assume that $F_w\approx G_{\theta_w}$ and $F_s\approx G_{\theta_s}$ are GEV distributions, even though the block size $b$ is small and still there might be some minor variability left within each season? If not, what kind of model should we choose instead? \\
As a compromise between parametric and purely non-parametric approaches we also present a regional estimator of high quantiles under semi-parametric model assumptions based on estimation of a regional extreme value index \citep{KinFriLil15}. For that purpose we are taking things a step backward: We avoid the GEV assumption by ignoring the fact that our observations $X$ are maximal values and simply treat them as realizations from any heavy-tailed distribution $F$. Instead of fitting the whole distribution $F$, we only focus on the estimation of an approximately parametric right tail. Similar to the two-component GEV model, this semi-parametric procedure leads us to more plausible estimates of high quantiles, since solely the largest observations from $F$ are taken effectively into account. \\
Our main results can be summarized as follows:
\begin{enumerate}
\item \textbf{Theory:} We discuss two new approaches for joint estimation of high quantiles of annual maximal flow distributions from a homogeneous group of $d$ stations: An estimator based on the two-component GEV model \citep{RulBuiRotKys16} and one based on semi-parametric assumptions. We derive asymptotic theory for both approaches under a very general framework (inter-site dependence, site-varying sample lengths) in order to be able to evaluate the estimation uncertainty of the procedures. This, for instance, leads to asymptotically optimal weights in the computation of regional estimates.
\item \textbf{Simulation:} Accounting for seasonality improves the estimation efficiency of high quantile estimators. If seasonal maximal flows, say, $W$ and $S$, are available, we recommend the parametric approach based on two-component distributions. Otherwise, if only annual data $X$ is available, the semi-parametric estimator still is rather competitive, provided the dimension $d$ of the homogeneous group is sufficiently large $(d\geq 5)$.
\item \textbf{Application:} We consider river flows from a homogeneous group of $d=8$ stations from the Mulde river basin, where winter and summer flows have considerably different tails. Our two-component GEV and semi-parametric procedures are able to deal with the seasonal variability, leading to more plausible estimates of high quantiles. \\
The severe flood in the summer of 2002 plays an outstanding role in our data set. The corresponding empirical return period/level combination is the only point in Figure \ref{fig:application} that is far away from the estimated curves. A classical GEV approach suggests that an event of at least the strength of 2002 occurs every 400 years on average. On the other hand, our two-component GEV approach suggests a return period of around 200 years, implying much higher risk of severe floods.
\end{enumerate}
The rest of the paper is organized as follows. We begin with a short description of the data in Section \ref{sec:data}. The next two sections introduce the models and related statistical methods, starting with two-component GEV distributions in Section \ref{sec:parametric} and followed by semi-parametric inference in Section \ref{sec:semi-parametric}. Section \ref{sec:sim} reports a simulation study and in Section \ref{sec:application} we apply the methods to maximal flows from several river stations located at the Mulde basin in Germany. Technical details and proofs are deferred to Appendix~\ref{sec:PWM} and \ref{sec:theory}.

\section{Data}\label{sec:data}

Regional frequency analysis (RFA) combines observations from many sites within a region, even though the primary interest is still the estimation of quantiles at one target site. \\
In our application, see Section \ref{sec:application}, we will have a look at a dataset which contains monthly maximal flows (in m$^3$/s) from $d=8$ stations located in the Mulde river basin (Saxony, Germany). The sample lengths of these sites range from 83 to 102 hydrological years (in Germany from November to October). The two seasons arise by splitting the hydrological year into the hydrological winter from November to April and summer from May to October; so each season contains 6 months. These two seasons are defined according to official guidelines \citep{dwa_merkblatt_2012}. The winter and summer maximal flows can be regarded as independent for any specific site, but simultaneous observations at sites from the same area are dependent. Within a site our data is complete, so we do not have to deal with missing values here. The structure of our data is described as follows: \\
Suppose that the random vector $\mathbf{X}=(X_1,\ldots,X_d)'$ represents annual maximal river flows at the $d$ stations, for any generic year. We denote the local distribution of site $j$ by $F_j(x)=\mathbb{P}(X_j\leq x)$. Unlike in many other regional approaches, we will not consider the components as independent. Let $\mathbf{X}_i=(X_{i,1},\ldots,X_{i,d})'$, with time index $i=1,\ldots,n$, be independent copies of $\mathbf{X}$, where $\{1,\ldots,n\}$ covers the observation period. 
In regional settings it is unlikely that the observation period is the same for all $d$ sites. Often the beginning of recordings is different for different sites because of the different ages of the measurement stations. If we arrange these $d$ sites according to their local sample lengths $n_j$, $j=1,\ldots,d$, such that $n=n_1\geq n_2\geq\ldots\geq n_d$, we obtain the observation scheme
\begin{align}\label{eq:sample} 
\left.\begin{array}{r}
X_{1,1},\ X_{2,1},\ X_{3,1}, \ X_{4,1}, \ X_{5,1}, \ \ldots,\ X_{n,1}~\sim F_1 \\
X_{a_2+1,2},\ X_{a_2+2,2},\  X_{a_2+3,2}, \ \ldots,\ X_{n,2}~\sim F_2 \\
\ddots\; \; \; \; \; \; \; \; \; \; \; \; \; \; \; \; \vdots \hspace{1.3cm}\\
X_{a_d+1,d},\ X_{a_d+2,d},\ \ldots,\ X_{n,d}~\sim F_d
\end{array}\right\}~\begin{array}{c}
\text{observations} \\
\text{from}~d~\text{sites}
\end{array}
\end{align}

where $a_j \in \{0,\ldots,n-1\}$ with $0=a_1\leq a_2 \leq \ldots \leq a_d$ and the local sample lengths $n_j=n-a_j$ for $j=1,\ldots,d$. Each row $X_{i,j}$, $i=a_j+1,\ldots,n,$ contains solely observations of one site $j\in\{1,\ldots,d\}$. It is important to consider scheme \eqref{eq:sample} in order to be able to take into account the dependence between local estimates computed from the scheme. \\
In fact, each of the observations from the previous scheme can be considered as a maximum $X_{i,j}=\max\{W_{i,j},S_{i,j}\}$ of corresponding winter and summer maximal flows, which is also relevant for the data set analyzed in Section \ref{sec:application}. That means, we have a scheme as in \eqref{eq:sample} for winter and summer maximal flows and combine them to the scheme of the annual maximal flows in \eqref{eq:sample}.

\section{Parametric estimation under seasonal variability}\label{sec:parametric}

In this section we assume that the available data consists of winter and summer maximal flows $W$ and $S$, with $W$ and $S$ being stochastically independent. Our focus is on estimation of high quantiles $q_p=F^{-1}(p)$ of the distribution $F(x)=\mathbb{P}(X\leq x)$ of an annual maximum $X=\max\{W,S\}$. By independence of $W$ and $S$, we obtain a factorization
\begin{align*}
F(x)=\mathbb{P}(\max\{W,S\}\leq x)=\mathbb{P}(W\leq x,~S\leq x)=\mathbb{P}(W\leq x)\cdot \mathbb{P}(S\leq x)
\end{align*}
and, because both $W$ and $S$ are maxima, we will also assume that these two variables are GEV distributed. We set $\mathbb{P}(W\leq x)=G_{\theta_w}(x)$ and $\mathbb{P}(S\leq x)=G_{\theta_s}(x)$ for some unknown seasonal parameters $\theta_w,\theta_s$ and with $G_\theta$ defined in \eqref{eq:GEV}. A distribution function $F=G_{\theta_w}\cdot G_{\theta_s}$ is called two-component GEV.

\subsection{Some remarks on two-component GEVs}

The parameters $\theta=(\mu,\sigma,\xi)$ of the GEV model $G_\theta$ defined in \eqref{eq:GEV} represent location, scale and shape, respectively. These distributions are unimodal and the shape $\xi\in\mathbb{R}$ determines the right tail behavior of $G_\theta$. The right tail is bounded, of exponential or of polynomial order for $\xi<0$, $\xi=0$ or $\xi>0$, respectively. The heaviness of the right tail, which may be considered as the risk of extraordinary large realizations from $G_\theta$, increases with increasing $\xi$. In extreme value theory this shape is also called extreme value index (EVI), denoted by $\gamma=\xi$. More generally, an EVI $\gamma$ can be determined for any distribution $F$ in the maximum domain of attraction. For instance, the normal distribution with its exponential tail behavior has an EVI of $\gamma=0$. Students $t_\nu$ and Fishers $F_{m,k}$ have a polynomial right tail with EVI $\gamma=1/\nu>0$ and $\gamma=2/k>0$, respectively. A more comprehensive list of parametric models and corresponding EVIs can be found in \citet[Chap. 2]{beirlant_statistics_2006}. \\
In case of positive shapes $\xi_w,\xi_s>0$, which is likely to hold in our data application, the EVI of $F=G_{\theta_w}\cdot G_{\theta_s}$ is given by $\gamma=\max\{\xi_w,\xi_s\}$, i.e., the right tail behavior of $F$ is determined by that of the heaviest component \citep[Cor. 1.11]{Kin16}. 

\smallskip

To fix ideas, suppose that our annual maximum $X$ follows a two-component GEV $F=G_{\theta_w}\cdot G_{\theta_s}$ with seasonal parameters $\theta_w=(\mu_w,\sigma_w,\xi_w)=(2,1,0.2)$ and $\theta_s=(\mu_s,\sigma_s,\xi_s)=(1.5,1,0.4)$. The corresponding density is depicted in Figure \ref{fig:dplot}. In practice we do not know the true model, so it might be preferable to forget seasonal variability and use instead a usual one-component GEV model $G_\theta$. Interestingly enough, by computing the minimizer
\begin{align*}
\theta^*=(\mu^*,\sigma^*,\xi^*)=\underset{\mu,\xi\in\mathbb{R},\sigma>0}{\operatorname{arg}\operatorname{min}}d_{KL}(F,G_\theta)=(2.554,~1.235,~0.305)
\end{align*}
\begin{figure}[h!]
	\centering
	\includegraphics[width=.85\linewidth]{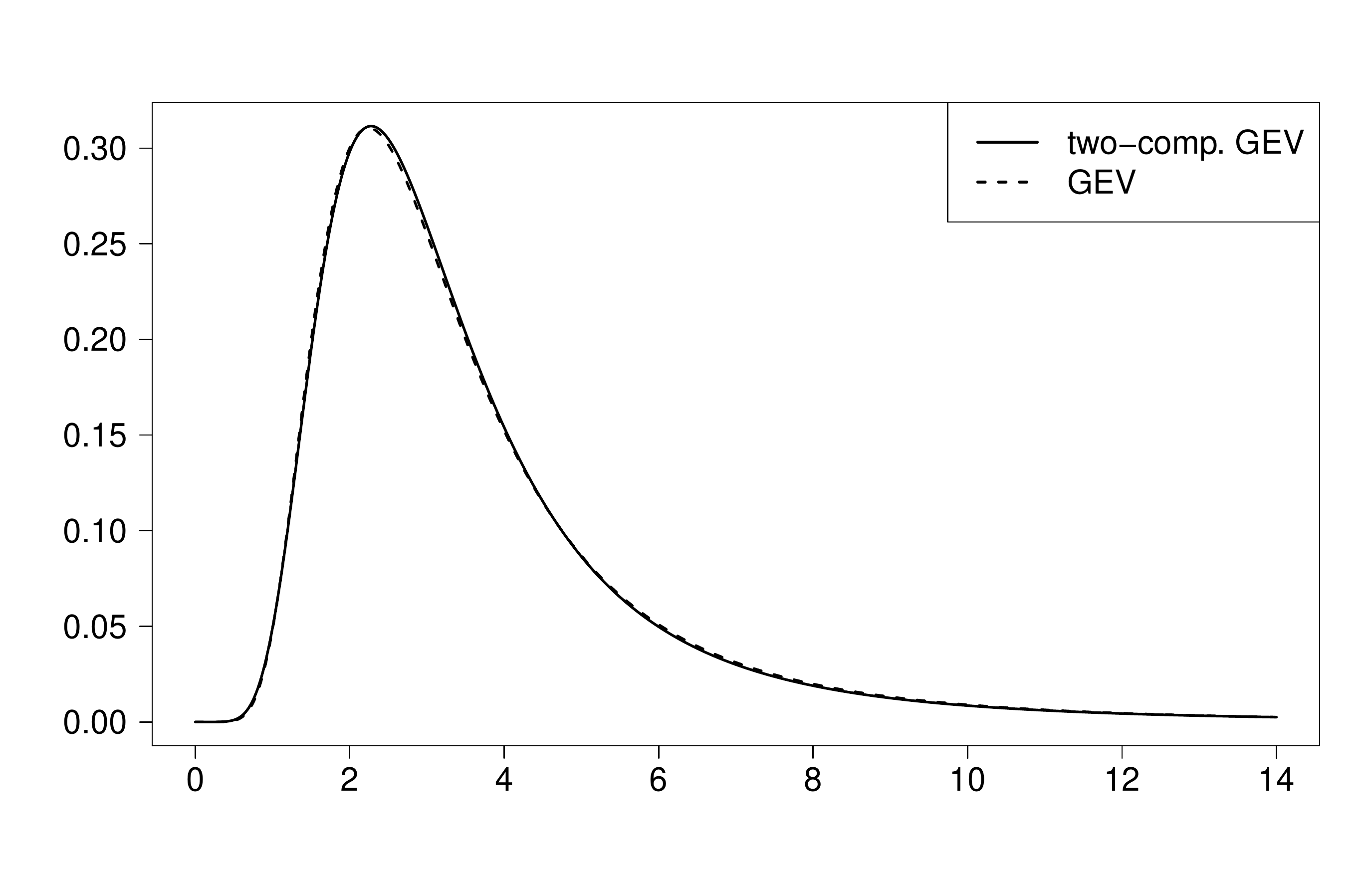}
	\caption{Density of a two-component GEV $F$ with $\theta_w=(2,1,0.2)'$ and $\theta_s=(1.5, 1, 0.4)'$, and that of a GEV with $\theta^*=(2.554, 1.235,0.305)'$ minimizing the Kullback-Leibler distance to $F$.}
	\label{fig:dplot}
\end{figure}
of the Kullback-Leibler distance $d_{KL}$ between $F$ and $G_\theta$, we obtain a GEV depicted in Figure \ref{fig:dplot} that looks almost identical to the density of $F$. Since we are mainly interested in high quantiles, we also compare these two models by their return level curves in Figure \ref{fig:rlplot}. From this second point of view we can clearly identify that these two models deviate reasonably from each other in the relevant tail region. The GEV model $G_{\theta^*}$ underestimates the true EVI of $\gamma=\max\{0.2,0.4\}=0.4$ by $\xi^*=0.305$, which results in underestimating risk of extreme realizations. A similar conclusion is drawn in our simulation study in Section \ref{sec:sim}.
\begin{figure}[h!]
	\centering
	\includegraphics[width=.85\linewidth]{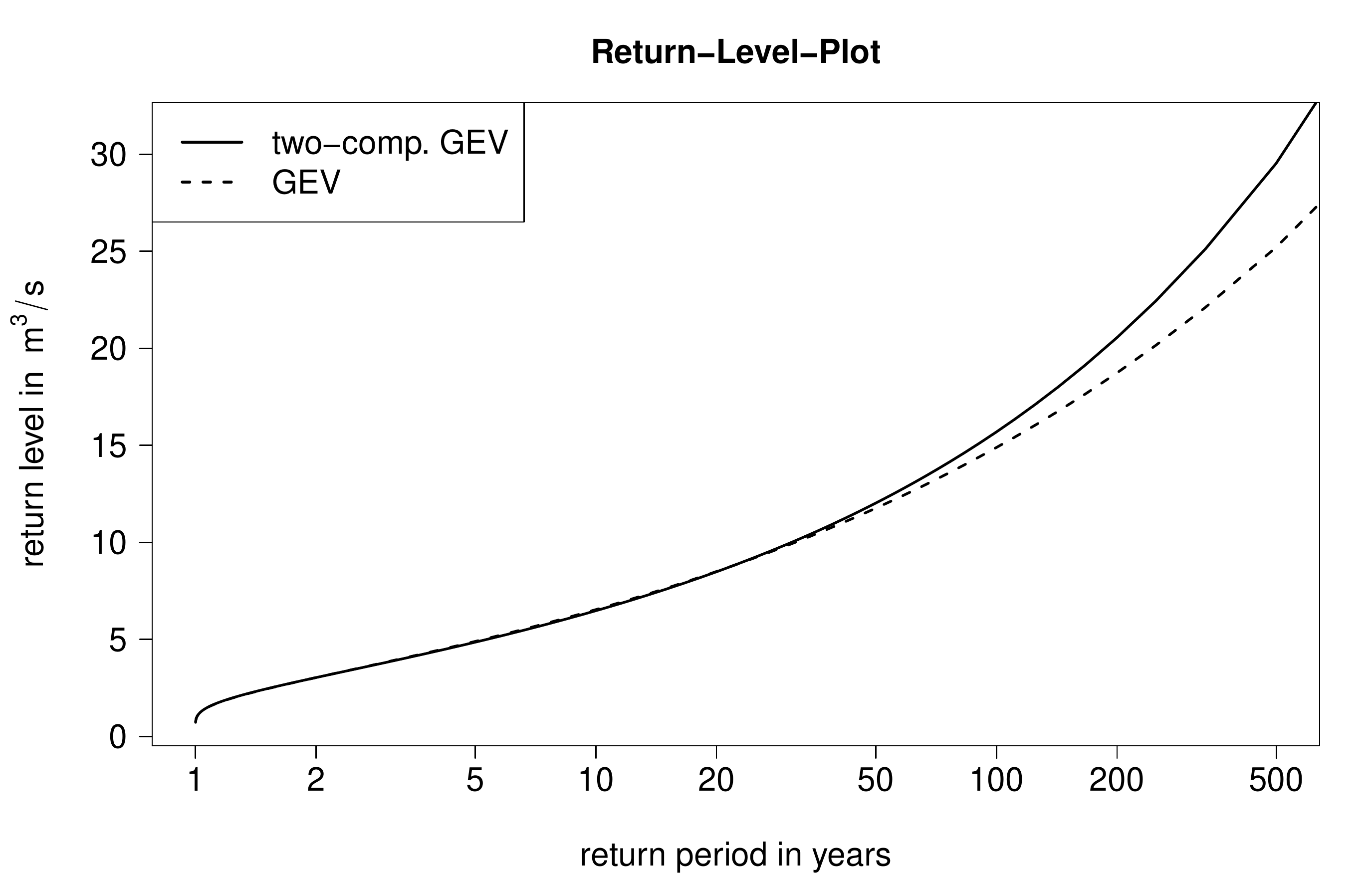}
	\caption{Return level plot of a two-component GEV $F$ with $\theta_w=(2,1,0.2)'$ and $\theta_s=(1.5, 1, 0.4)'$, and that of a GEV with $\theta^*=(2.554, 1.235,0.305)'$ minimizing the Kullback-Leibler distance to $F$.}
	\label{fig:rlplot}
\end{figure}

\subsection{Accuracy of quantile estimates}

Let $F=G_{\theta_w}\cdot G_{\theta_s}$ be a two-component GEV distribution, where $\theta_w,\theta_s$ are two unknown parameter vectors and where $G_\theta$ is defined in \eqref{eq:GEV}. For fixed $p\in(0,1)$ we consider the problem of estimating
\begin{align}
q_p=F^{-1}(p)~\text{ by }~\hat{q}_p=\hat{q}_{p,n}=\left(G_{\hat{\theta}_w}\cdot G_{\hat{\theta}_s}\right)^{-1}(p),
\end{align}
where $\hat{\theta}_w$ and $\hat{\theta}_s$ are estimators of $\theta_w$ and $\theta_s$, respectively. It should be noted that the equation $q=\varphi(p)=\left(G_{\theta_w}\cdot G_{\theta_s}\right)^{-1}(p)$ does not have an explicit solution, except in case of $\theta_w=\theta_s$. Since $\varphi$ is monotone increasing, it is easy to obtain an accurate numerical solution. However, a non-trivial difficulty arises when one is not only interested in a point-estimate $\hat{q}_p$ but also in its precision or in a confidence interval. \\
Suppose that $\hat{\theta}_w=\hat{\theta}_{w,n}$ and $\hat{\theta}_s=\hat{\theta}_{s,n}$ are two independent, asymptotically normal and consistent estimators with asymptotic covariance matrices
\begin{align*}
\Sigma_w=\lim_{n\rightarrow\infty}{\rm Var}\left[\sqrt{n}\big(\hat{\theta}_w-\theta_w\big)\right]~\text{ and }~\Sigma_s=\lim_{n\rightarrow\infty}{\rm Var}\left[\sqrt{n}\big(\hat{\theta}_s-\theta_s\big)\right].
\end{align*}
In our application we consider the regional TL-moment estimators from Appendix \ref{sec:pwm}, with $\hat{\theta}_w$ computed from winter and $\hat{\theta}_s$ from summer maximal flows. The corresponding theory, the asymptotic normality of the regional TL-moment approaches, has been proved in a recent paper by \citet{LilKinFri16}. \\
In this paper we prove that $\hat{q}_p$ is an asymptotically normal and consistent estimator of $q_p$ (see Appendix \ref{sec:seasonal_parametric}) with limiting variance
\begin{align}\label{eq:two_comp_var}
\sigma^2_p=\lim_{n\rightarrow\infty}{\rm Var}\left[\sqrt{n}\big(\hat{q}_p-q_p\big)\right]=\frac{G^2_{\theta_s}(q_p)\cdot J_{\theta_w}(q_p)\Sigma_w J_{\theta_w}(q_p)'+G^2_{\theta_w}(q_p)\cdot J_{\theta_s}(q_p)\Sigma_s J_{\theta_s}(q_p)'}{\left[g_{\theta_w}(q_p)\cdot G_{\theta_s}(q_p)+G_{\theta_w}(q_p)\cdot g_{\theta_s}(q_p)\right]^2},
\end{align}
where $g_\theta$ is the density of $G_\theta$ and $J_\theta(x)=\frac{\partial}{\partial \theta}G_\theta(x)\in\mathbb{R}^{1\times 3}$ is the Jacobi matrix of $G_\theta(x)$ in $\theta$. \\
In practice all unknown quantities on the right-hand side of \eqref{eq:two_comp_var} can be replaced by their sample counterparts in order to obtain an estimate $\hat{\sigma}^2_p$ of $\sigma^2_p$. This, in turn, is used to calculate an asymptotically valid $(1-\alpha)$-confidence interval
\begin{align}\label{eq:confinterval_two-comp}
I_{1-\alpha}=\left[\hat{q}_p - z_{1-\alpha/2}\cdot\frac{\hat{\sigma}_p}{\sqrt{n}},~\hat{q}_p + z_{1-\alpha/2}\cdot\frac{\hat{\sigma}_p}{\sqrt{n}}\right]
\end{align}
for $\hat{q}_p$ and $\alpha\in(0,1)$, where $z_{1-\alpha/2}$ denotes the $(1-\alpha/2)$-quantile of $\mathcal{N}(0,1)$.

\section{Semi-parametric estimation of heavy tails}\label{sec:semi-parametric}

Contrary to the previous section, we only require annual observations $X$ for the computation of estimates in this section. In addition, we will not assume that the annual distribution $F$ of $X$ is determined by a finite number of parameters. Instead we apply a semi-parametric framework under a certain heavy-tail assumption: A distribution function $F$ with right-unlimited support is called a Pareto-type distribution function with extreme value index $\gamma>0$, if
\begin{align}\label{eq:paretotypedist}
\bar{F}(x)=1-F(x)=x^{-1/\gamma}\cdot L(x),\ x>0,
\end{align}
for some (measurable) function $L:\mathbb{R}_+\rightarrow\mathbb{R}_+$ satisfying $L(tx)/L(t)\rightarrow1$ for $t\rightarrow\infty$ and all $x>0$. The characterization in \eqref{eq:paretotypedist} is called semi-parametric, since the parameter $\gamma>0$ is of core interest, while the distribution cannot be characterized by a finite-dimensional parameter vector due to the presence of the function $L$. Typical textbook examples of $L$ are $L(x)=\log(1+x)$ or simply a constant function $L\equiv c>0$. Roughly speaking, slowly varying functions can be considered as almost constant $L(x)\approx c$ for values $x>u$ above some large threshold $u$, which implies that the right tail of $F$ behaves like a power function, $\bar{F}(x)\approx c\cdot x^{-1/\gamma}$ for $x>u$. Some examples of Pareto-type distributions with extreme value index $\gamma>0$ are Student's $t_\nu$ with $\gamma=1/\nu$, Fisher's $F_{m,k}$ with $\gamma=2/k$, the $\log$-Gamma$(\lambda,\alpha)$ with $\gamma=1/\lambda$, the GEV$(\mu,\sigma,\xi)$ and the generalized Pareto GP$(\sigma,\xi)$ with $\gamma=\xi>0$. \\
Recall that our interest is in estimation of some high quantile $q_p=F^{-1}(p)$, where we often have to deal with $p\geq1-1/n$. The Pareto-type framework allows us to estimate $F$ beyond the range of observations without relying on a parametric model: Let $U(t)=F^{-1}(1-1/t)$, $t>1$, denote the $t$-return level of $F$. It is known that Pareto-type distributions $F$ with extreme value index $\gamma>0$ satisfy
\begin{align}\label{eq:Uchar}
\lim_{t\rightarrow\infty}\frac{U(tx)}{U(t)} = x^\gamma~\text{ for all }~x>0.
\end{align}
We thus have $U(tx)\approx U(t)\cdot x^\gamma$ for large $t$, which means that a high quantile $q_p=U(tx)$ is approximately determined by a moderate quantile $U(t)$ and the index $\gamma$. \\
In practice one usually sets $t<n$, with $n$ denoting the sample length, such that $U(t)$ is estimated with satisfactory precision from a simple order statistic. The difficulty remains particularly with the estimation of $\gamma$.

\subsection{Semi-parametric Index Flood}\label{sec:semi-IF}

Recently, \citet{DemCle15} and \citet{KinFriLil15} studied the Hill estimator of $\gamma$ \citep{Hil75} under a regional tail homogeneity assumption. Their theory allows us to introduce a new approach what we call a \textsl{Semi-parametric Index Flood} procedure. We will use the regional estimator of $\gamma$ from the previous references and plug it in into the so-called Weissman extrapolation formula \citep{Wei78} for semi-parametric estimation of high quantiles. For readers who are not familiar with semi-parametric heavy-tail analysis, Hill's estimator and Weissman's extrapolation formula, we recommend to read the introduction into these topics provided in Appendix \ref{sec:unimath} first. \\
Suppose that our data consists of annual maximal river flows from $d$ different sites as in \eqref{eq:sample} and that we are interested in the estimation of, say, $F^{-1}_1(p)$. We assume that each marginal distribution is of Pareto-type,
\begin{align}\label{eq:pareto-type-regional}
F_j(x)=\mathbb{P}(X_{i,j}\leq x)=1-x^{-1/\gamma_j}L_j(x),~x>0,~j=1,\ldots,d,
\end{align}
with not necessarily identical slowly varying functions $L_j$ but with the same unknown extreme value index $\gamma>0$, that is,
\begin{align}\label{eq:tailhomogeneity}
\mathcal{H}_{0,\text{evi}}:\ \gamma_1=\ldots=\gamma_d=\gamma\ \text{ for some }\ \gamma>0.
\end{align}
We call assumption \eqref{eq:tailhomogeneity} heavy-tail homogeneity. Note that $\mathcal{H}_{0,\text{evi}}$ in combination with the Pareto-type framework \eqref{eq:pareto-type-regional} is much weaker than the Index Flood assumption stated in \eqref{eq:IFassumption}. We do not require the distributions $F_j$ to be parametric and we only need equality in EVI $\gamma$, which in the GEV framework means that we only need that the shape parameters $\xi_1=\ldots=\xi_d=\gamma>0$ are identical and positive. \\
Let $\hat{u}_j$ denote the $(n_j-k_j)$-th largest observation and $\hat{\gamma}_j$ denote Hill's estimator computed from the $k_j$ relative excesses above the threshold $u_j$ from sample $X_{a_j+1,j},X_{a_j+2,j},\ldots,X_{n,j}$, $k_j \in \{2,\ldots,n_j \}$.  Taking advantage of assumption \eqref{eq:tailhomogeneity}, we set
\begin{align}\label{eq:jointquantile}
\hat{F}_1^{-1}(p)=\hat{u}_1\cdot \left(\frac{k_1}{n_1(1-p)}\right)^{\hat{\gamma}(\bm{w})}~\text{ with }~\hat{\gamma}(\bm{w})=\sum_{j=1}^dw_j\hat{\gamma}_j,
\end{align}
where $\bm{w}=(w_1,\ldots,w_d)'$ is a vector of weights summing up to 1. The left-hand side of \eqref{eq:jointquantile}, also known as Weissman's extrapolation formula \citep{Wei78}, is motivated by \eqref{eq:Uchar} with $t=n_1/k_1$ and $x=k_1/(n_1(1-p))$. \\
In conformity with equation \eqref{eq:IFassumption}, the estimator \eqref{eq:jointquantile} can be viewed as a semi-parametric Index Flood approach, with site-specific scales $\hat{u}_j$ and regional tail behavior determined by $\hat{\gamma}$. The probability of large events $x>\hat{u}_1$ is estimated by the inversion of \eqref{eq:jointquantile},
\begin{eqnarray}\label{eq:tailapproxireg}
\hat{F}_1(x)=1-\frac{k_1}{n_1}\cdot \left(\frac{x}{\hat{u}_1}\right)^{-1/\hat{\gamma}(\mathbf{w})}.
\end{eqnarray}

\subsection{Selection of weights and tail samples}\label{sec:select_k_and_w}

We recap the rules considered in \citet{KinFriLil15} and begin with the selection of weights $\bm{w}$. For the moment, let $\mathbf{k}=(k_1,\ldots,k_d)'$ be a fixed vector of integers. The joint estimator $\hat{\gamma}(\bm{w})$ on the right-hand side of \eqref{eq:jointquantile} is asymptotically normal with
\begin{align}
{\rm Var}\left[\hat{\gamma}(\bm{w})\right]\approx \frac{\gamma^2}{k_1}\bm{w}'\Sigma\bm{w} \nonumber
\end{align}
and $\Sigma\in\mathbb{R}^{d\times d}$ defined in Appendix \ref{sec:semi_theory}. Using Lagrange multipliers it is straightforward to show that
\begin{align}
\bm{w}_{\text{opt}}=\underset{\bm{w}\in\mathbf{W}}{\operatorname{arg}\operatorname{min}}\frac{\gamma^2}{k_1}\bm{w}'\Sigma\bm{w}=\left(\mathbf{1}'\Sigma^{-1}\mathbf{1}\right)^{-1}\cdot\Sigma^{-1}\mathbf{1},
\end{align}
where $\mathbf{1}=(1,\ldots,1)'\in\mathbb{R}^d$ and $\mathbf{W}$ is the set of all vectors in $\mathbb{R}^d$ with components summing up to 1. Thanks to Proposition \ref{prop:weissmanconverg}, $\bm{w}_{\text{opt}}$ is also the minimizer of the limiting variance of the high quantile estimator from \eqref{eq:tailapproxireg}. \\
Let us now turn to the selection of $\mathbf{k}$. The numbers $k_j$, $j=1,\ldots,d$, represent the local tail sample lengths that are effectively used in the estimation of the common EVI $\gamma$. These numbers typically are small relative to the full sample lengths so that the Hill estimator and related procedures should be applied only in data-rich situations. Since our local sample lengths $n_j$ are rather short, we counteract this local drawback by combining many different stations from a region. \\
The fact that we deal with block maxima from $d$ different stations motivates us to set $k_j=\lfloor 2n_j^{2/3}/d^{1/3}\rfloor$ \citep[Sec. 3.2]{KinFriLil15}. Note that the ratio $k_j/n_j$ is getting smaller with increasing local sample size $n_j$. This is needed for the consistency of Hill's estimator. Accounting for the dimension $d$ in $k_j$ turns out to be important in our data applications. We are able to reduce a typically dominant bias by taking advantage of a large dimension $d$.

\section{Simulation study}\label{sec:sim}

We choose two numbers of sites $d\in\{5,10\}$ in order to consider common regional settings. For simplicity, we assume that all $d$ sites have the same observation period and hence the same sample lengths $n\in\{50,100\}$, which means that $a_j=0$ in \eqref{eq:sample} for all $j=1,\ldots,d$, even though our theory is able to deal with the general case of different sample lengths. The focus lies on the (regional) estimation of a high quantile $q_p=F_j^{-1}(p)$ for probability $p\in\{0.99,0.999\}$ at some site, say, $j=1$.  \\
We apply Khoudraji's device \citep{Kho95} in order to generate an inter-site dependence model. We set
\begin{align}
C_{\bm{\theta},\mathbf{c}}=C_{\theta_1}(\mathbf{u}^\mathbf{c})\cdot C_{\theta_2}(\mathbf{u}^{1-\mathbf{c}}),~ \bm{\theta}=(\theta_1,\theta_2)\in[1,\infty)^2,~\mathbf{c}\in[0,1]^d,
\end{align}
where $C_\theta$ denotes the $d$-dimensional Gumbel-Hougaard copula, $\bm{\theta}=(1.5, 2.5)'$ accounts for strength of dependence and $\mathbf{c}=1/d \cdot(0,1,2,\ldots,d-1)'$ for the amount of asymmetry. This model is simple to handle but still allows for some sort of dependence asymmetry. The latter being a typical feature of discharge data. \\
We compare the following high quantile estimators computed from annual maxima $X$:
\begin{enumerate}
	\item[-] Weissman's estimator \eqref{eq:jointquantile} with $\mathbf{k}$ and $\bm{w}$ according to Section \ref{sec:select_k_and_w}, \hfill (W)
	\item[-] $L$-moment estimator from Appendix \ref{sec:PWM} and \hfill (L)
	\item[-] Trimmed $L$-moment estimator from Appendix \ref{sec:PWM}.\hfill (TL)
\end{enumerate}
(L) and (TL) are improved regional estimators based on the method of probability weighted moments proposed by \citet{LilKinFri16}. After combining the estimation methods, the sample lengths and the number of sites, we get different scenarios to look at. Each scenario is replicated 5000 times. 

\smallskip

Recall that in the introduction we identified two potential sources of deviations of $F$ from a GEV distribution: Finite block length $b$ and seasonal variability in general. The next two subsections are devoted to each of these sources.


\subsection{Block maxima observations}

\begin{figure}[t!]
	\centering
	\includegraphics[width=.85\linewidth]{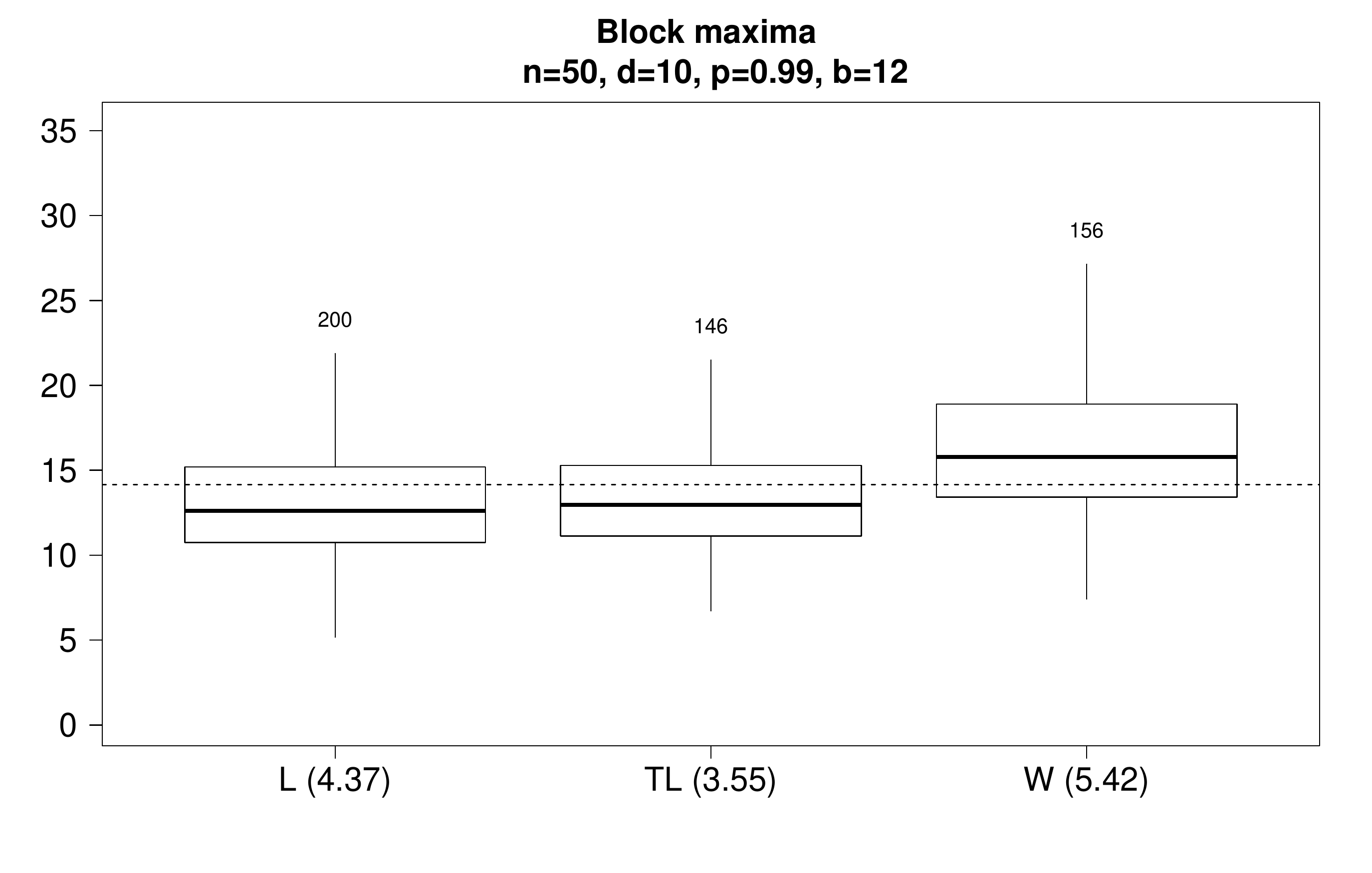}
	\caption{Boxplots of the 99\%-quantile estimations obtained from 5000 replications with block maxima margins $F_j^{(b)}$ from \eqref{eq:block_max_dist}, $b=12$, $n=50$ and $d=10$. The dashed line marks the true quantile of the block maxima distribution $q_{0.99}=14.151$. Outliers are not shown, but the numbers of such values are depicted at the end of the whiskers. The mean squared error divided by $q_{0.99}^2$ for scaling is shown in brackets.}
	\label{fig:blockmax100_5}
\end{figure}

Annual maximal flows are computed by maximization over finite blocks of, say, monthly observations and thus, it is common to assume that annual maxima are exactly GEV distributed. This section analyzes the behavior of the estimators when observations are in fact generated by block maxima distributions. More precisely, let $T_\nu$ denote the $t$-distribution function with $\nu$ degrees of freedom. We choose
\begin{align}\label{eq:block_max_dist}
F_j(x)=F_j^{(b)}(x)=\left[2\cdot T_{1/\xi}\left(\left\{1+\xi\frac{x-\mu_j}{\sigma_j}\right\}\cdot T_{1/\xi}^{-1}\left(1-\frac{1}{2b}\right)\right)-1\right]^b,
\end{align}
where in general $b\geq2$. It can be shown that $F_j^{(b)}$ is the distribution function of a block maximum
\begin{align*}
\mu_j+\frac{\sigma_j}{\xi}\left(\frac{\max\{|Z_1|,\ldots,|Z_b|\}}{a_b}-1\right)\quad \text{with } a_b=T_{1/\xi}^{-1}\left(1-\frac{1}{2b}\right),
\end{align*}
over independent and identically $t_{1/\xi}$-distributed random variables $Z_1,\ldots,Z_b$. The standardization is used in order to obtain $\lim_{b\rightarrow\infty}F_j^{(b)}=G_{\mu_j,\sigma_j,\xi}$, which is the GEV distribution function with location $\mu_j$, scale $\sigma_j$ and shape $\xi$. \\
Here, we have a look at annual maximal flows, so we set $b=12$. For the parameters we choose $(\mu_j, \sigma_j, \xi)' = (1.75, 1, 0.3)'$ for all $j=1,\ldots, d.$

\smallskip

For the different scenarios, i.e., dimensions $d$, sample lengths $n$ and probabilities $p$, we came to the same conclusion, so just one particular scenario with $d=10$, $n=50$ and $p=0.99$ is depicted in Figure \ref{fig:blockmax100_5}. The true 99\%-quantile of this scenario is $q_{0.99}=14.151 \text{ m}^3/\text{s}$, which corresponds to the dashed line in the figure. All competitors show a similar bias, caused by finite $b$ for the parametric estimators (model misspecification) and by finite $n$ for the Weissman approach. We complete the analysis of finite blocks $b$ and continue with seasonal observations under the GEV assumption.

\subsection{Seasonal observations}

\begin{figure}[t!]
	\centering
	\includegraphics[width=.85\linewidth]{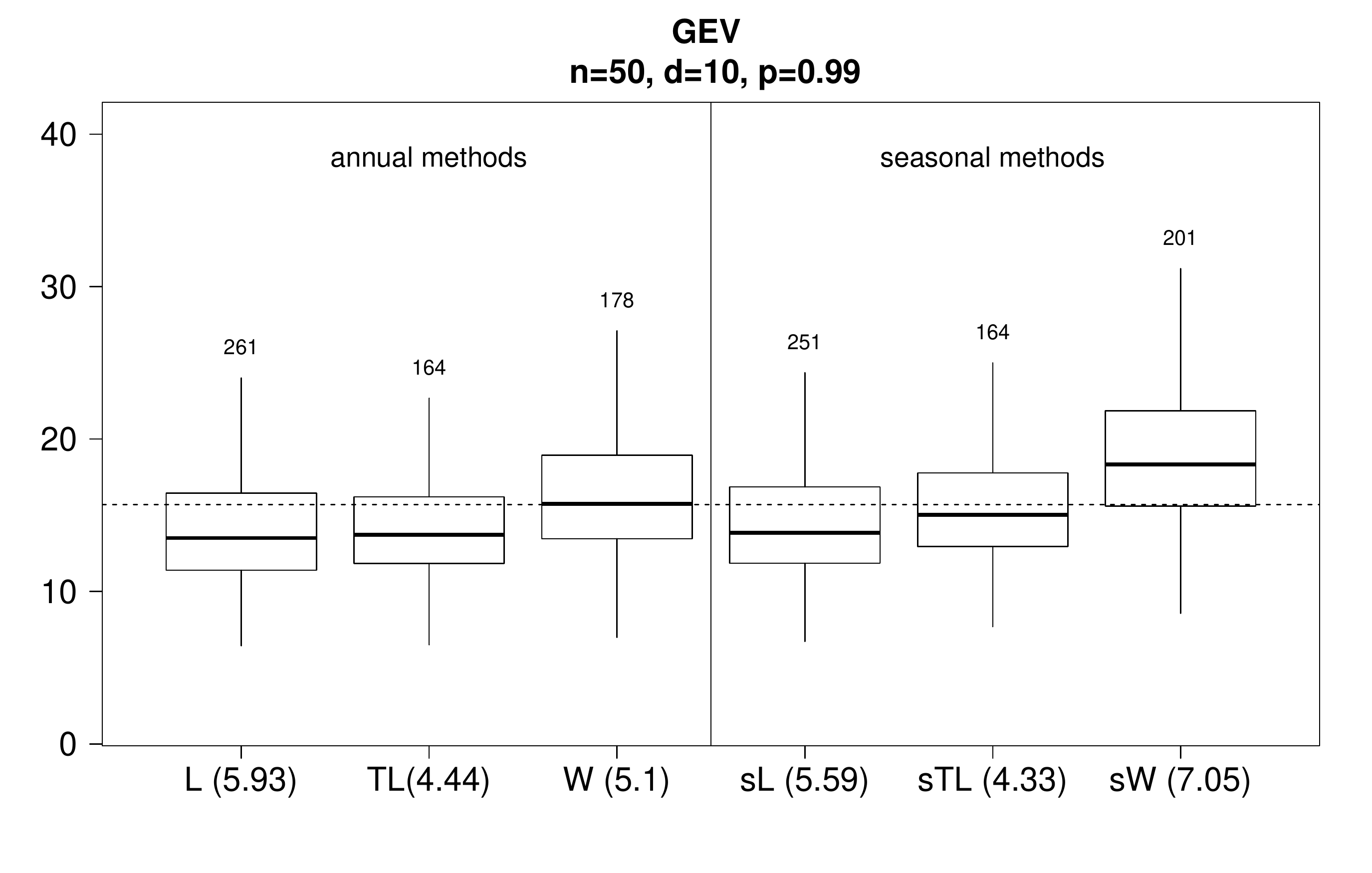}
	\caption{Boxplots of the 99\%-quantile estimations of 5000 replications for a seasonal scenario with $n=50$ and $d=10$. The dashed line marks the true quantile of the product of the two GEVs $q_{0.99}= 15.692$. Outliers are not shown, but the numbers are depicted at the end of the whiskers. The mean squared error divided by $q_{0.99}^2$ for scaling is shown in brackets.}
	\label{fig:gev100_5}
\end{figure}

Now, according to \eqref{eq:productdist} and the related discussion, we set $F_j=F_{w,j}\cdot F_{s,j}$, where each of the seasonal distributions is a GEV with parameters $\mu_{w,j},\sigma_{w,j},\xi_{w,j}$ and $\mu_{s,j},\sigma_{s,j},\xi_{s,j}$. Here, besides the methods (W), (L) and (TL) used before, we also apply seasonal estimators from Section \ref{sec:seasonal_parametric}, namely
\begin{enumerate}
\item[-] seasonal Weissman's estimator, \hfill (sW)
\item[-] seasonal $L$-moments and \hfill (sL)
\item[-] seasonal Trimmed $L$-moments. \hfill (sTL)
\end{enumerate}

Based on experience from real datasets we choose the parameters $(\mu_{w,j},\sigma_{w,j},\xi_{w,j})'=(2,1,0.2)'$ for the winter and $(\mu_{s,j},\sigma_{s,j},\xi_{s,j})'=(1.5, 1, 0.4)'$ for the summer season. The true quantile in this scenario is $q_{0.99}= 15.692 \text{ m}^3/\text{s}$. \\
Figure \ref{fig:gev100_5} shows the boxplots for the regional scenario with $d=10$, $n=50$ and $p=0.99$. We omit the results for other combinations of $(p,n,d)$ because the conclusions drawn are similar. \\
The annual L- and TL-estimators are based on misspecified assumptions, since here we consider observations from two-component GEV's. In our simulation experiments this is reflected by a non-negligible bias. Annual Weissman (W) and seasonal TL-moments (sTL) outperform their competitors, probably because they do not suffer from the previous misspecification issue. However, this is also true for seasonal Weissman (sW), and still its performance suffers from a devastating bias. We are not able to explain this poor behavior and can neither prove the presumed asymptotic normality of estimator sW. \\
Summing up, we do not recommend sW. We neither recommend L and TL if strong seasonal variability is present. Instead, we suggest sTL if one assumes that variability is not an issue within each season. Otherwise, we advise to apply W, provided that the dimension $d$ is not very small. Note that local estimation (not reported here) with the Weissman approach can be not recommended due to the typically dominant bias. 

\subsection*{A comment on alternative scenarios}

We also studied several other scenarios that are not reported here because of qualitatively very similar conclusions. For instance, we studied the performance under a variety of plausible parameter values for $(\mu_{w,j},\sigma_{w,j},\xi_{w,j})'$, $(\mu_{s,j},\sigma_{s,j},\xi_{s,j})'$, $(\mu_j, \sigma_j, \xi)'$ and local sample lengths $n_j$, where plausibility is based on our personal experience from real data applications. We also looked at scenarios where both sources of deviations of $F$ from a GEV model, seasonal variability and block maxima observations for each season, have been used. Summing up, we always drew conclusions very similar to that of the experiments reported here.

\section{Application: The Mulde river basin} \label{sec:application}

\begin{figure}[b!]
	\centering
	\includegraphics[width=.85\linewidth]{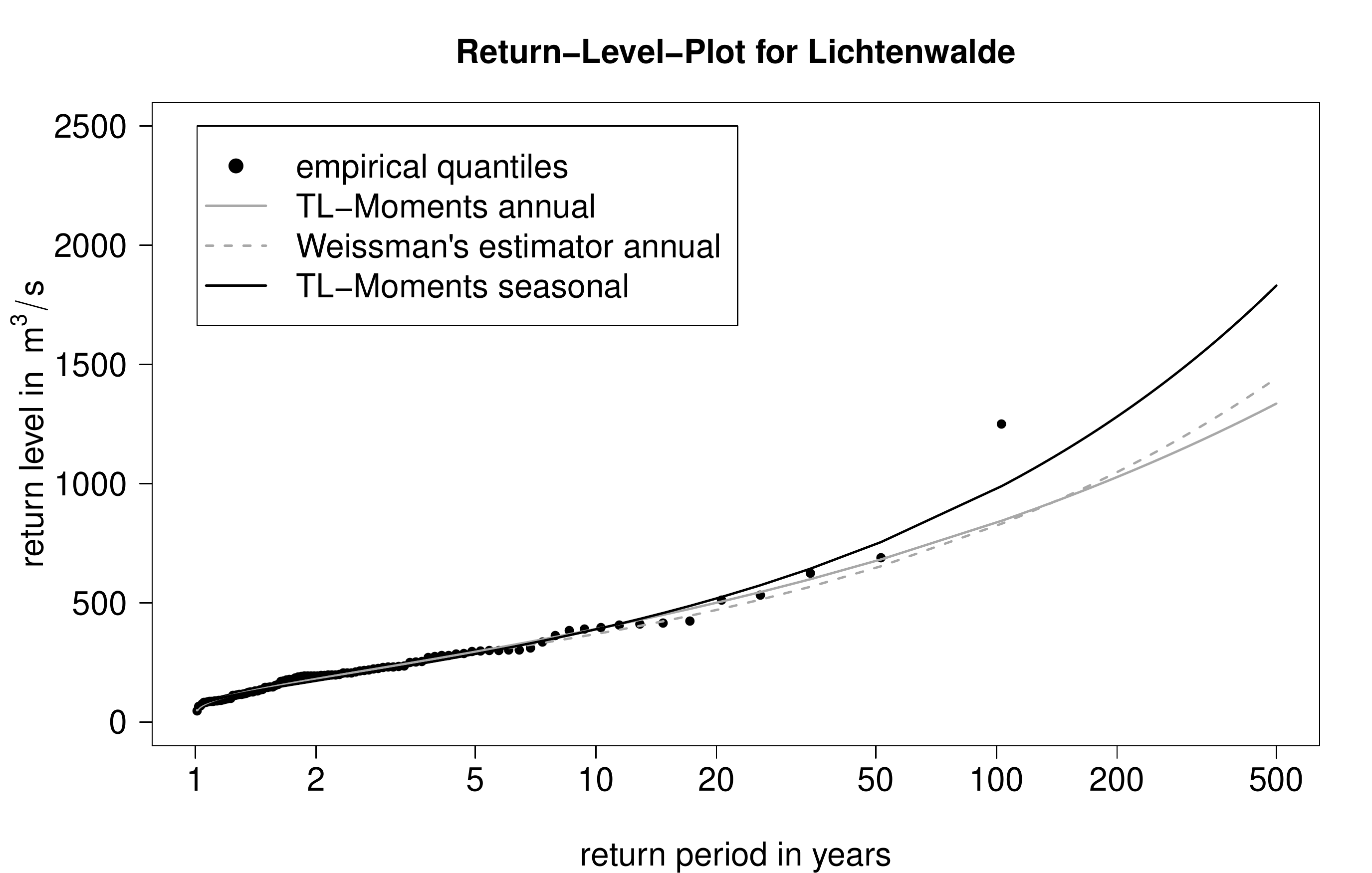}
	\caption{Return level plot for the site Lichtenwalde with sample lengths $n=102$.}
	\label{fig:application}
\end{figure}

The Mulde river basin with a catchment area of about 7400\,km$^2$ is located in Saxony, a federal state in the East of Germany. The basin is named after a main tributary of the Elbe river, the ``Vereinigte Mulde'', with an average flow at the mouth of about 73\,m$^3$/s. From a flood frequency analysts point of view, this region is particularly interesting because there were two flood disasters in the recent past. The first one was in 2002 and the second one in 2013. Both were summer events caused by heavy rainfall during some days. The damage caused amounts to several billion euros. It is thus of particular relevance to account for risk adequately. \\
We study a data set of monthly maximal flows from $d=8$ sites: Lichtenwalde, Golzern, Zickau-Pölbitz, Wechselburg, Nossen, Hopfgarten, Pockau and Borstendorf. Their local sample lengths $n_j, j=1,\ldots,d,$ are varying between 83 and 102 years. A hydrological year runs from November to October and typical seasons in this region are winter from November to April and summer from May to October. We calculate winter, summer and annual maximal flows $W_{i,j}$, $S_{i,j}$ and $X_{i,j}=\max\{W_{i,j},S_{i,j}\}$ for each available combination $(i,j)$ of year and station resulting in schemes as in \eqref{eq:sample}. In order to apply our regional estimators to the annual (resp. seasonal) observations, it is recommended to check first assumption $\mathcal{H}_{0,\text{tail}}$ from \eqref{eq:tailhomogeneity} for annual (resp. seasonal) data. The application of homogeneity test based on the main result of \citet{LilKinFri16} considering the hypothesis of identical shape parameters resulted in $p$-values of 0.41, 0.63 and 0.54 for annual, winter and summer schemes, respectively. This indicates that there are no serious sources of tail heterogeneity. In what follows, since regional estimation works comparably well also under moderate heterogeneity \citep{LetWalWoo87}, we consider the sites as tail homogeneous. \\
Figure \ref{fig:application} depicts a return level plot for station Lichtenwalde, which is one of the eight sites of the group. In general, a point in the plot has coordinates $\big(T,F^{-1}(1-1/T)\big)$, with $T>1$ called return period and $F^{-1}(1-1/T)$ called $T$-year return level. The three lines in Figure \ref{fig:application} correspond to three regional estimators of $F$: The Weissman estimator, the TL-moment estimator under the usual GEV assumption and the sTL-moment estimator under the two-component GEV assumption. The first two take only annual flows $X_{i,j}$ into account, while the third one is based on seasonal observations $W_{i,j}$ and $S_{i,j}$. These three procedures turn out to be superior to other considered competitors (see Section \ref{sec:sim}). The dots $\big(T_{i:n},X_{i:n}\big)$, $i=1,\ldots,n$, in the figure represent the empirical counterpart: The ordered observations $X_{1:n}\leq\ldots\leq X_{n:n}$ are called empirical return levels with corresponding empirical return periods $T_{i:n}=1/(1-p_{i:n})$ and plotting positions $p_{i:n}=i/(n+1)$. \\
In Figure \ref{fig:application} we observe that estimated $T$-year return levels of the three competing estimators are pretty much the same up to $T\approx 40$ years. This corresponds to almost the same quantile curves $\hat{F}^{-1}(p)$ for $p\leq 1-1/40=0.975$. Note that the largest observation stands out from the bulk of the data. The corresponding empirical return level of around 100 years is half of the sTL and quarter of both annual estimated return levels. The fact that the seasonal approach is closer to the sample is also the case in many other of our data applications (not reported here). This, indeed, should not be always interpreted as an advantage, because of possible data-overfitting. For instance, the empirical version $F_n$ perfectly fits the data even though this estimator is not recommended due to its horrible variability in the tails. However, note also that seasonal estimation is based on extended information: We do not simply take annual $X_{i,j}$ but use variables $W_{i,j}$, $S_{i,j}$ instead and account for possible seasonal variability. In this particular application the variability in the tails is reflected by very different estimates of the shapes, $\hat{\xi}_w=0.29$ for winter and $\hat{\xi}_s=0.41$ for the summer season. Without accounting for the variability we would estimate the shape as $\hat{\xi}_{\text{annual}}=0.25$, which surprisingly is not in-between the first two values. This counter-intuitive result is explained by the fact that mixing two very different distributions can result in misleading overlapping. \\
In the end, practitioners are interested in the estimation of, say, the $p=99\%$-quantile or, equivalently, the $T=100$ year return level of the annual maximal flow distribution $F$ of Lichtenwalde. Here we focus on the regional two-component GEV approach from Appendix \ref{sec:pwm} based on TL-moments. Our theory allows us to calculate a point estimate $\hat{q}_p^{(\text{sTL})}=962.5$ and its corresponding estimated 95\%-confidence interval of $\hat{I}_{0.95}^{(\text{sTL})}=[680.9,~1244.1]$ by plugging in estimates of all unknown quantities in formula \eqref{eq:confinterval_two-comp}. For comparison, the point estimate and 95\%-confidence interval for the local sTL (resp. regional W) procedure is given by $\hat{q}_p^{(\text{loc sTL})}=1089.6$ and $\hat{I}_{0.95}^{(\text{loc sTL})}=[568.0,~1611.2]$ (resp. $\hat{q}_p^{(\text{W})}=823.2$ and $\hat{I}_{0.95}^{(\text{W})}=[521.0,~1125.4]$), where local means that only observations from Lichtenwalde are used for estimation of the two-component GEV model.


\section{Conclusion}

We consider asymptotic methods for regional frequency analysis that are able to deal with seasonal variability and inter-site dependence. Our framework is flexible in the sense that we do allow for very different local sample lengths and we do not restrict ourselves to parametric dependence models. \\
The methods are investigated via simulations and illustrated by a real data application. Taking into account possible seasonal variability can have an enormous impact on the results. For instance, the return period of the severe flood in 2002 estimated by the two-component approach was only half of that by the classical one-component GEV approach. We recommend seasonal models whenever the corresponding observations are available in the given data set. \\
It even might be useful to consider more than $K=2$ seasons. A corresponding theoretical result, the asymptotic normality for $K$-component GEV quantiles, can be easily deduced along the same lines as the proof of case $K=2$.

\section*{Acknowledgements}

We are grateful to Professor Andreas Schumann from the Department of Civil Engineering, Ruhr-University Bochum, Germany, for providing us hydrological data and for helpful discussions. The financial support of the Deutsche Forschungsgemeinschaft (SFB 823, ``Statistical modelling of nonlinear dynamic processes'') is gratefully acknowledged.

\appendix

\section{GEV parameterizations by probability weighted moments}\label{sec:PWM}

Let $G_\theta$  with parameters $\theta=(\mu,\sigma,\xi)'$ denote a GEV distribution function defined in \eqref{eq:GEV}. Throughout this section we will assume that $\xi<1$. The probability weighted moment (PWM) of $G_\theta$ of order $k\in\mathbb{N}_0$ is defined by
\begin{align}\label{eq:pwm}
\beta_k=\beta_k(\theta)=\int_\mathbb{R} x\cdot G_\theta^k(x)~dG_\theta(x)=\int_{(0,1)} G_\theta^{-1}(u)\cdot u^k~du.
\end{align}
It is a well-known fact that every distribution with finite expectation is uniquely determined by its sequence of PWM's \citep{Hos07}. It is thus little surprising that the parameter $\theta$ already is determined by a finite number of PWM's. Two such re-parameterizations of the GEV family are recapped below.

\smallskip

\textbf{1) Method of L-moments}

\smallskip

The most popular re-parameterization of the GEV based on the first three PWM's $\bm{\beta}=(\beta_0,\beta_1,\beta_2)'$ is implicitly given by the equation system
\begin{align}\label{eq:Lmom_equations}
\left\{
\begin{array}{cll}
\frac{3^\xi-1}{2^\xi-1}&=\frac{3\beta_2-\beta_0}{2\beta_1-\beta_0}&=\frac{\lambda_3}{2\lambda_2}+\frac{3}{2} \\
\sigma&=\frac{(2\beta_1-\beta_0)\xi}{\Gamma(1-\xi)(2^\xi-1)}&=\frac{\lambda_2\xi}{\Gamma(1-\xi)(2^\xi-1)} \\
\mu&=\beta_0+\frac{\sigma}{\xi}\left[1-\Gamma(1-\xi)\right]&=\lambda_1+\frac{\sigma}{\xi}\left[1-\Gamma(1-\xi)\right]
\end{array}
\right. ,
\end{align}
where $\Gamma$ denotes the gamma function and with $\lambda_1=\beta_0$, $\lambda_2=2\beta_1-\beta_0$, $\lambda_3=6\beta_2-6\beta_1+\beta_0$ denoting the first three so-called L-moments of $G_\theta$ \citep{Hos90}. For some reasons to be found in the previous reference, practitioners usually prefer to work with L-moments instead of PWM's. \\
However, an explicit solution $\phi$ with $\theta=\phi(\bm{\beta})$ of the previous equation system does not exist. Instead of relying on a numerical solver, it is convenient to replace the first equation in \eqref{eq:Lmom_equations} by the approximation
\begin{align}
\xi=-7.859\cdot h(\bm{\beta})-2.9554\cdot h(\bm{\beta})^2~\text{ with }~h(\bm{\beta})=\frac{2\beta_1-\beta_0}{3\beta_2-\beta_0}-\frac{\log 2}{\log 3}=\frac{2}{3+\lambda_3/\lambda_2}-\frac{\log 2}{\log 3}
\end{align}
in order to obtain an explicit solution. The error in $\xi$ caused by the approximation is known to be smaller than 0.0009 for $-1/2<\xi</1/2$ \citep{HosWalWoo85} and is therefore negligible in practice. \\

\medskip

\textbf{2) Method of TL-moments}

\smallskip

\citet{ElaSeh03} proposed a generalization of L-moments denoted by $\lambda_k^{(r,s)}$, where $(r,s)\in\mathbb{N}_0^2$ is a trimming parameter: Larger values for $r$ (resp. $s$) decrease the influence of the smallest (resp. largest) observations, with $\lambda_k^{(0,0)}=\lambda_k$ being the usual (untrimmed) L-moments. For further insight we refer to \citet{Hos07}. \\
For the reminder we focus on the trimming $(r,s)=(0,1)$, which means that only the influence of the largest observations is slightly reduced and which has proven to be a good choice in our simulation experiments. Again, the procedure for the derivation of GEV parameters will be based only on the first three trimmed L-moments and, thanks to \citet[Sec. 2.4]{Hos07}, we have that $\lambda_1^{(0,1)}=2\beta_0-2\beta_1$, $\lambda_2^{(0,1)}=3/2(4\beta_1-\beta_0-3\beta_2)$ and $\lambda_3^{(0,1)}=2/3(36\beta_2-18\beta_1+2\beta_0-20\beta_3)$. It can be verified that the equation system
\begin{align}\label{eq:TLmom_equations}
\left\{
\begin{array}{cll}
\frac{5\cdot 4^\xi-12\cdot 3^\xi+9\cdot 2^\xi-2}{3^\xi-2^{\xi+1}+1}&=\frac{2\cdot(18\beta_2-9\beta_1+\beta_0-10\beta_3)}{4\beta_1-\beta_0-3\beta_2}&=\frac{9\lambda_3^{(0,1)}}{4\lambda_2^{(0,1)}} \\
\sigma&=\frac{4\beta_1-\beta_0-3\beta_2}{\Gamma(-\xi)\cdot(3^\xi-2^{\xi+1}+1)}&=\frac{2/3\cdot \lambda_2^{(0,1)}}{\Gamma(-\xi)\cdot(3^\xi-2^{\xi+1}+1)} \\
\mu&=2(\beta_0-\beta_1)+\frac{\sigma}{\xi}-\frac{\sigma\cdot\Gamma(-\xi)}{(2^\xi-2)^{-1}}&=\lambda_1^{(0,1)}+\frac{\sigma}{\xi}-\frac{\sigma\cdot\Gamma(-\xi)}{(2^\xi-2)^{-1}} \\
\end{array}
\right. 
\end{align}
determines the GEV parameters. Similar to the first procedure, we replace the first equation in \eqref{eq:TLmom_equations} by the approximation
\begin{align*}
\xi=-8.567394\cdot h(\bm{\beta})+0.675969\cdot h(\bm{\beta})^2~\text{ with }~h(\bm{\beta})=\frac{6\beta_1-3/2\beta_0-9/2\beta_2}{15\beta_2-5/3\beta_0-40/3\beta_3}-\frac{2\log 2-\log 3}{3\log 3-2\log 4}
\end{align*}
in order to obtain an explicit solution, where now $\bm{\beta}=(\beta_0,\beta_1,\beta_2,\beta_3)'$ denotes the vector of the first four PWM's of $G_\theta$.

\section{Asymptotic statistics}\label{sec:theory}

\subsection{A CLT for PWMs with applications to GEV parameter estimation}\label{sec:pwm}

We recap the main result from \citet{LilKinFri16}. These authors verify the joint asymptotic normality of sample PWMs
\begin{align*}
\hat{\beta}_{k,j,r_j,n}=\frac{1}{n_j}\sum_{i=1}^{n_j}X_{a_j+i}\cdot F^k_{j,a_j+1:n}(X_{a_j+i}),~k=0,\ldots,K,~j=1,\ldots,d
\end{align*}
computed from scheme \eqref{eq:sample}.

\begin{theorem}\label{theo:pwm}
Suppose that $\mathbf{X}_i$, $i\geq1$, is a sequence of independent copies of $\mathbf{X}=(X_1,\ldots,X_d)'$ with marginal distribution functions $F_j(x)=\mathbb{P}(X_j\leq x)$ and with marginal PWMs
\begin{align*}
\bm{\beta}=\left(\beta_0(F_1),\ldots,\beta_K(F_1),\beta_0(F_2),\ldots,\beta_K(F_2),\ldots\ldots,\beta_0(F_d),\ldots,\beta_K(F_d)\right)'\in\mathbb{R}^{d\cdot K}.
\end{align*}
We further assume that
\begin{align*}
\mathbb{E}\left[X_{j}F_j^k(X_{j})X_{\ell}F_\ell^m(X_{\ell})\right]<\infty~\text{ for all }~1\leq j,\ell\leq d~\text{ and }~0\leq k,m< K
\end{align*}
and that $\sup_{x\in\mathbb{R}}\left|x\{F_j(x)(1-F_j(x))\}^w\right|<\infty$ for all $j=1,\ldots,d$ and some $w\in[0,1/2)$. Let $\hat{\bm{\beta}}_{\bm{r},n}$ denote the sample counterpart of $\bm{\beta}$ computed from scheme \eqref{eq:sample} with $n_j/n\rightarrow r_j\in\mathbb{R}_+$ for $n\rightarrow\infty$. Then, for $n\rightarrow\infty$, we have that
\begin{align}
\sqrt{n}\left(\hat{\bm{\beta}}_{\bm{r},n}-\bm{\beta}\right)\stackrel{D}{\longrightarrow}\mathcal{N}\left(0,~\Sigma_{\bm{r}}\right),
\end{align}
where the limiting variance matrix $\Sigma_{\bm{r}}\in\mathbb{R}^{dK\times dK}$ is defined below.
\end{theorem}

\citet{LilKinFri16} proved that the matrix $\Sigma_{\bm{r}}=\lim_{n\rightarrow\infty}{\rm Var}\left(\sqrt{n}\left(\hat{\bm{\beta}}_{\bm{r},n}-\bm{\beta}\right)\right)$ is defined block-wise by
\begin{align*}
\lim_{n\rightarrow\infty}{\rm Cov}\left(\sqrt{n}\left(\hat{\bm{\beta}}_{j,r_j,n}-\bm{\beta}_j\right),~\sqrt{n}\left(\hat{\bm{\beta}}_{\ell,r_\ell,n}-\bm{\beta}_\ell\right)\right)=\frac{\min(r_j,r_\ell)}{r_j\cdot r_\ell}\cdot{\rm Cov}(\mathbf{Z}_j,~\mathbf{Z}_\ell)\in\mathbb{R}^{K\times K},
\end{align*}
where $\bm{\beta}_j=\left(\beta_0(F_j),\ldots,\beta_K(F_j)\right)'$ with sample counterpart $\hat{\bm{\beta}}_{j,r_j,n}$ computed from the $j$-th row of scheme \ref{eq:sample}. The random vectors $\mathbf{Z}_j=(Z_{0,j},Z_{1,j},\ldots,Z_{K-1},j)'$, $j=1,\ldots,d$, are defined componentwise by \begin{align}\label{eq:ztrue}
Z_{k,j}=X_j\cdot F_j^k(X_j)+\int_{\mathbb{R}}x\cdot k\cdot F_j^{k-1}(x)\cdot \mathbbm{1}(X_j\leq x)~dF_j(x).
\end{align}
The matrices ${\rm Cov}(\mathbf{Z}_j,\mathbf{Z}_\ell)$ are consistently estimated by their sample analogues: Let
\begin{align}\label{eq:zhat}
\hat{Z}_{i,k,j}=X_{i,j}\cdot F_{j,a_j+1:n}^k(X_{i,j})+\frac{1}{n_j}\sum_{\ell=1}^{n_j}X_{\ell,j}\cdot k\cdot F_{j,a_j+1:n}^{k-1}(X_{\ell,j})\cdot \mathbbm{1}(X_{i,j}\leq X_{\ell,j})
\end{align}
and $\hat{\mathbf{Z}}_{i,j}=(Z_{i,0,j},Z_{i,1,j},\ldots,Z_{i,K-1,j})'$, $i=a_j+1,\ldots,n$. For $1\leq j,\ell,\leq d$, the covariance matrix ${\rm Cov}(\mathbf{Z}_j,\mathbf{Z}_\ell)$ is estimated by the empirical covariance matrix of the sample
\begin{align*}
\left\{\left(\hat{\mathbf{Z}}_{\max(a_j,a_\ell)+1,j},\hat{\mathbf{Z}}_{\max(a_j,a_\ell)+1,\ell}\right),\ldots,\left(\hat{\mathbf{Z}}_{n,j},\hat{\mathbf{Z}}_{n,\ell}\right)\right\}.
\end{align*}
This empirical estimator of $\Sigma_{\bm{r}}$ is denoted by $\hat{\Sigma}_{\bm{r},n}$.

\smallskip

Note that the preceding theorem is purely non-parametric, i.e., except for continuity and the existence of some moments no further assumptions on the distribution are made. For the reminder of this section we consider a parametric problem, the estimation of GEV parameters under a regional homogeneity assumption: Suppose now that $F_j=G_{\mu_j,\sigma_j,\xi_j}$ is a GEV distribution function for all $j=1,\ldots,d$ and that the tail homogeneity assumption
\begin{align}\label{eq:H0_tail}
\mathcal{H}_{0,\text{tail}}:~\xi_1=\ldots=\xi_d=\xi~\text{ for some }~\xi<1/2
\end{align}
holds. Let $\hat{\theta}_{j,r_j,n}=(\hat{\mu}_j,\hat{\sigma}_j,\hat{\xi}_j)'$ denote either the L- or TL-moment estimator of $\theta_j=(\mu_j,\sigma_j,\xi_j)'$ computed from the $j$-th row of scheme \eqref{eq:sample}, with computation based on sample PWMs and the corresponding representation from Appendix \ref{sec:PWM}. Let $\bm{w}\in\mathbb{R}^d$ denote any vector of weights summing up to 1. Thanks to assumption \eqref{eq:H0_tail}, we can estimate $\theta_j$ consistently by the regional version $\big(\hat{\mu}_j,\hat{\sigma}_j,\bm{w}'\hat{\bm{\xi}}\big)$. The asymptotic normality and the corresponding limiting variance can be easily derived from the previous theorem and by the delta method. \\
Just like for the regional Weissman estimator it is meaningful to set
\begin{align*}
\bm{w}_{\text{tail}}=\underset{\bm{w}\in\mathbf{W}}{\operatorname{arg}\operatorname{min}}\bm{w}'\Sigma_{\text{tail}}\bm{w}=\left(\mathbf{1}'\Sigma_{\text{tail}}^{-1}\mathbf{1}\right)^{-1}\cdot\Sigma_{\text{tail}}^{-1}\mathbf{1},
\end{align*}
where $\Sigma_{\text{tail}}=\lim_{n\rightarrow\infty}{\rm Var}\left[\sqrt{n}(\hat{\bm{\xi}}-\bm{\xi})\right]$ with $\hat{\bm{\xi}}=(\hat{\xi}_1,\ldots,\hat{\xi}_d)'$. Again, $\Sigma_{\text{tail}}$ is obtained by Theorem \ref{theo:pwm} and the delta method. Finally, we set $\hat{\xi}=\hat{\bm{w}}_{\text{tail}}'\hat{\bm{\xi}}$ as our regional (T)L-moment estimator of the regional shape $\xi$, where $\Sigma_{\text{tail}}$ involved in $\bm{w}_{\text{tail}}$ is replaced by a consistent estimator $\hat{\Sigma}_{\text{tail}}$.  \\
\citet{LilKinFri16} introduced also a check version $\check{\Sigma}_{\bm{r},n}$ for the estimation of $\Sigma_{\bm{r}}$ working under the GEV assumption. $\check{\Sigma}_{\bm{r},n}$ is a hybrid version, consisting of parametric $3\times 3$-blocks on the diagonal and non-parametric ones outside the diagonal. Throughout the simulation section we applied the check-version because of its superior efficiency. However, in finite samples this approach sometimes leads to invalid covariance matrices reflected by negative eigenvalues. In such cases we simply replace $\bm{w}_{\text{tail}}$ by weights $\bm{w}_{\text{ind}}=(n_1,n_2,\ldots,n_d)'/\sum_{j=1}^d n_j$ proportional to local sample lengths $n_j$, which is the optimal choice under the assumption of spatial independence.

\subsection{Quantiles of two-component GEV distributions}\label{sec:seasonal_parametric}

Let $F(x)=G_{\theta_w}(x)\cdot G_{\theta_s}(x)$ be a two-component GEV distribution, where $\theta_w,\theta_s$ are two unknown parameter vectors and where $G_\theta$ is defined in \eqref{eq:GEV}. Let $p\in(0,1)$ be fixed. Suppose that our interest is in the estimation of
\begin{align}
q_p=F^{-1}(p)~\text{ by the estimator }~\hat{q}_p=\hat{q}_{p,n}=\left(G_{\hat{\theta}_w}\cdot G_{\hat{\theta}_s}\right)^{-1}(p)
\end{align}
and in assessing the estimation uncertainty of $\hat{q}_p$. We will further assume that $\hat{\theta}_w=\hat{\theta}_{w,n}$ and $\hat{\theta}_s=\hat{\theta}_{s,n}$ are independent satisfying
\begin{align*}
\sqrt{n}\left(\hat{\theta}_w-\theta_w\right)\stackrel{D}{\longrightarrow}\mathcal{N}\left(0,~\Sigma_w\right)~\text{ and }~\sqrt{n}\left(\hat{\theta}_s-\theta_s\right)\stackrel{D}{\longrightarrow}\mathcal{N}\left(0,~\Sigma_s\right)
\end{align*}
for $n\rightarrow\infty$.

\begin{theorem}
Let $g_\theta$ denote the density of $G_\theta$ and $J_\theta(x)=\frac{\partial}{\partial\theta}G_\theta(x)\in\mathbb{R}^{1\times 3}$ the Jacobi matrix of $G_\theta(x)$ in $\theta$. Then, for $n\rightarrow\infty$, we have that
\begin{align}
\sqrt{n}\left(\hat{q}_p-q_p\right)\stackrel{D}{\longrightarrow}\mathcal{N}\left(0,~\sigma^2_{p;\theta_w,\theta_s}\right),
\end{align}
where the limiting variance is given by
\begin{align}
\sigma^2_{p;\theta_w,\theta_s}=\frac{G^2_{\theta_s}(q_p)\cdot J_{\theta_w}(q_p)\Sigma_w J_{\theta_w}(q_p)'+G^2_{\theta_w}(q_p)\cdot J_{\theta_s}(q_p)\Sigma_s J_{\theta_s}(q_p)'}{\left[g_{\theta_w}(q_p)\cdot G_{\theta_s}(q_p)+G_{\theta_w}(q_p)\cdot g_{\theta_s}(q_p)\right]^2}.
\end{align}
\end{theorem}
\begin{proof}
Let $p\in(0,1)$ be fixed and let $\varphi:\Theta\times\Theta\subset\mathbb{R}^6\rightarrow\mathbb{R}$ denote the map $(\theta_1,\theta_2)\mapsto \left(G_{\theta_1}\cdot G_{\theta_2}\right)^{-1}(p)$. We are going to prove that $\varphi$ is a differentiable map in $(\theta_w,\theta_s)$. \\
Let $J$ denote the support of $F$ and let $I\subset J$ be some compact interval with interior point $q_p\in I$. Note that $\varphi$ can be decomposed into $\varphi=\psi\circ\eta$ with
\begin{align*}
\eta:\Theta\times\Theta\rightarrow\ell^\infty(I),~(\theta_1,\theta_2)\mapsto G_{\theta_1}\cdot G_{\theta_2}~\text{ and }~\psi:\ell^\infty(I)\rightarrow\mathbb{R},~g\mapsto g^{-1}(p).
\end{align*}
By Lemma \ref{lem:hadamard_cdf} below and a simple multivariate extension of \citet[Lem. 12.2]{Kos08}, it is straightforward to verify that $\eta$ is Hadamard differentiable in $(\theta_w,\theta_s)$ with derivative $D\eta_{(\theta_w,\theta_s)}:\mathbb{R}^6\rightarrow \ell^\infty(I)$,
\begin{align*}
\mathbf{a}\mapsto\left(G_{\theta_s}\cdot J_{\theta_w},~G_{\theta_w}\cdot J_{\theta_s}\right)\times \mathbf{a},~\mathbf{a}\in \mathbb{R}^6,
\end{align*}
where $J(\theta,x)=\frac{\partial}{\partial\theta}G_\theta(x)\in\mathbb{R}^{1\times 3}$ is the Jacobi matrix of $G_\theta(x)$ in $\theta$ and where $\times$ denotes matrix multiplication. By \citet[Lem. 3.9.20]{Van2000}, $\psi$ is Hadamard differentiable in $F=\eta(\theta_w,\theta_s)\in\ell^\infty(I)$ with derivative $D\psi_F:\ell^\infty(I)\rightarrow\mathbb{R}$,
\begin{align*}
h\mapsto-\frac{h(F^{-1}(p))}{f(F^{-1}(p))},~h\in\ell^\infty(I),
\end{align*}
where $f\in\ell^\infty(I)$ is the derivative of $F$. \\
Since both, $\eta$ and $\psi$, are Hadamard differentiable, so is the composition $\varphi=\psi\circ\eta$. But this is equivalent to $\varphi$ being differentiable in the usual sense \citep{Van2000}. The chain rule is also valid for Hadamard differentiability, which implies that the Jacobi matrix $D\varphi_{(\theta_w,\theta_s)}$ is given by
\begin{align*}
D\varphi_{(\theta_w,\theta_s)}=D\psi_{F}\circ D\eta_{(\theta_w,\theta_s)}=-\frac{\left(G_{\theta_s}(q_p)\cdot J(\theta_w,q_p),~G_{\theta_w}(q_p)\cdot J(\theta_s,q_p)\right)}{g_{\theta_w}(q_p)\cdot G_{\theta_s}(q_p)+G_{\theta_w}(q_p)\cdot g_{\theta_s}(q_p)}\in\mathbb{R}^{1\times 6}.
\end{align*}
Recall that $F=\eta(\theta_w,\theta_s)$, $q_p=F^{-1}(p)$ and $f(x)=g_{\theta_w}(x)\cdot G_{\theta_s}(x)+G_{\theta_w}(x)\cdot g_{\theta_s}(x)$. By the delta method and for $n\rightarrow\infty$, we finally conclude that
\begin{align*}
\sqrt{n}(\hat{q}_p-q_p)=\sqrt{n}\left(\varphi(\hat{\theta}_w,\hat{\theta}_s)-\varphi(\theta_w,\theta_s)\right)\stackrel{D}{\longrightarrow}\mathcal{N}\left(0,~\sigma^2_{p;\theta_w,\theta_s}\right),
\end{align*}
where
\begin{align*}
\sigma^2_{p;\theta_w,\theta_s}=D\varphi_{(\theta_w,\theta_s)}\cdot\begin{pmatrix}
\Sigma_w & \\
 & \Sigma_s
\end{pmatrix}\cdot (D\varphi_{(\theta_w,\theta_s)})'.
\end{align*}
\end{proof}

\begin{lemma}\label{lem:hadamard_cdf}
Let $I_{0}$ denote the support of $G_{\theta_0}$ and let $I\subset I_0$ be some compact interval. The map $\kappa:\Theta\rightarrow\ell^\infty(I)$, $\theta\mapsto G_\theta$, is Hadamard-differentiable in $\theta_0$ with derivative $D\kappa_{\theta_0}:\Theta\rightarrow\ell^\infty(I)$, $\mathbf{a}\mapsto J_{\theta_0}\mathbf{a}$, where $J_\theta(x)=\frac{\partial}{\partial\theta}G_\theta(x)$, $x\in I$.
\end{lemma}
\begin{proof}
Hadamard differentiability is defined as a particular kind of uniform convergence of the difference quotient of $\kappa$. We have to prove that 
\begin{align}\label{eq:uniform_conv_diff}
\left\|\frac{\kappa(\theta_0+t\mathbf{a}_t)-\kappa(\theta_0)}{t}-D\kappa_{\theta_0}(\mathbf{a})\right\|\longrightarrow 0~\text{ for }~t\rightarrow 0
\end{align}
and any $\theta_0+t\mathbf{a}_t\in\Theta$ with $\mathbf{a}_t\rightarrow\mathbf{a}$ for $t\rightarrow 0$. We consider the norm $\|f\|=\sup_{x\in I}|f(x)|$ on $\ell^\infty(I)$. It is a known fact that uniform convergence on compact sets follows by point-wise convergence and by equicontinuity. \\ 
Note that point-wise convergence, i.e., convergence in \eqref{eq:uniform_conv_diff} for fixed $x\in I$ follows by the differentiability of the map $\theta\mapsto G_\theta(x)$ in $\theta_0$. Since $I$ is a compact set, it suffices to prove equicontinuity of the family $\{f_t\in\ell^\infty(I):~ t\in[-1,0)\cup(0,1]\}$ with
\begin{align*}
f_t(x)=\frac{G_{\theta_0+t\mathbf{a}_t}(x)-G_{\theta_0}(x)}{t},~~x\in I.
\end{align*}
\textsl{Proof of equicontinuity:} Let $\varepsilon>0$ and $x_1,x_2\in I$ with $|x_1-x_2|<\delta$ ($\delta$ to be defined later on). By the mean value theorem there exists a real number $x_0$ between $x_1$ and $x_2$, such that
\begin{align*}
|f_t(x_1)-f_t(x_2)|=\left|f'_t(x_0)\right|\cdot |x_1-x_2|,
\end{align*}
where $f'_t(x)=\frac{\partial}{\partial x}f_t(x)=t^{-1}\left[g_{\theta_0+t\mathbf{a}_t}(x)-g_{\theta_0}(x)\right]$ and $g_\theta$ is the density of $G_\theta$. If we show that $\left|f'_t(x_0)\right|\leq c$ is bounded uniformly for all $x_0\in I$ and $t\in[-1,1]\backslash\{0\}$, we can set $\delta=\varepsilon/c$ and the proof of the equicontinuity is complete. \\
\textsl{Proof of boundedness:} Note first that the density $g_\theta$ with $\theta=(\mu,\sigma,\xi)\in\Theta$ is bounded by $|g_\theta(x)|\leq\sigma^{-1}|1+\xi|$. Without loss of generality we may assume that $\theta_0+t\mathbf{a}_t$ is bounded uniformly in $t\in[-1,1]\backslash\{0\}$, since $\mathbf{a}_t\rightarrow\mathbf{a}$ for $t\rightarrow0$. This allows us to find a bound for $|f'_t(x)|$ uniformly in $t\in[-1,-a]\cup[a,1]$ and $x\in I$, where $0<a<1$ is some constant. On the other hand (for $t$ close to zero), note that $\lim_{t\rightarrow 0}f'_t(x)=\frac{\partial}{\partial\theta}g_\theta(x)|_{\theta=\theta_0}\mathbf{a}$ and that the limit is uniformly bounded for $x\in I$. Summing up, we conclude that $|f'_t(x)|$ is bounded uniformly for $t\in(-a, a)\backslash\{0\}$. This completes the proof of the lemma.
\end{proof}

\subsection{Semi-parametric inference}\label{sec:unimath}

\subsubsection{Introduction: Hill's estimator and Weissman's extrapolation formula}

Let $X$ be a random variable with distribution function $F(x)=\mathbb{P}(X\leq x)$ and let $u>0$ be a real positive number. The random variable $X/u$ satisfying $X>u$ is called relative excess over the threshold $u$. Then, from \eqref{eq:paretotypedist}, we obtain
\begin{align}
\mathbb{P}\left(\frac{X}{u}\leq x~\Big|~ X>u\right)=1-\frac{\bar{F}(ux)}{\bar{F}(u)}\rightarrow 1-x^{-1/\gamma}\ \text{ for }\ u\rightarrow\infty, \nonumber
\end{align}
which means that relative excesses over $u$ approximately follow a parametric distribution for large $u$. The limit $P_\gamma(x)=1-x^{-1/\gamma}$, $x\geq1$, is called Pareto distribution function with $\gamma>0$ being a shape parameter. It is straightforward to verify that a random variable $Z$ with distribution function $P_\gamma$ fulfills $\gamma=\mathbb{E}[\log Z]$, which suggests to estimate $\gamma$ from arithmetic means of $\log$-transformed excesses: \\
Let $X_1,\ldots,X_n$ be independent copies of $X$ with ordered values denoted by $X_{(1)}\leq\ldots\leq X_{(n)}$. Hill's estimator of $\gamma$ proposed in \citet{Hil75} is defined by
\begin{align}\label{eq:hillest}
H_{k,n}=\frac{1}{k}\sum_{i=1}^k\log\left(\frac{X_{(n-i+1)}}{u_{k,n}}\right)\ \text{with}\ u_{k,n}=X_{(n-k)},
\end{align}
for integers $1\leq k<n$ representing the number of excesses over the random threshold $u_{k,n}$. The popularity of Hill's estimator can be explained by its computational simplicity and by the fact that it can be considered as a maximum pseudo-likelihood estimator. \\
Now, we choose $\hat{\gamma}=H_{k,n}$ as Hill's estimator. In the next step, we estimate quantiles $F^{-1}(p)$ with an extrapolation formula of \citet{Wei78}, which is motivated by relation \eqref{eq:Uchar} and given by 
\begin{align}\label{eq:weissmanextrapolation}
\hat{F}^{-1}(p)=u_{k,n}\cdot\left(\frac{k}{n(1-p)}\right)^{\hat{\gamma}}.
\end{align}
The first term $u_{k,n}=X_{(n-k)}$ is interpreted as a non-parametric estimator of a moderate quantile $F^{-1}(1-k/n)=U(n/k)$, while the remainder is used for the extrapolation into the tail region $p>1-k/n$, with $\hat{\gamma}$ controlling the extrapolation width.

\subsubsection{Limit theorems for regional semi-parametric estimation}\label{sec:semi_theory}

The following results are built upon technical assumptions, that, however, can be checked only if detailed information on the tail region of the joint distribution is available. Since our observation lengths are very limited, we omit these details and refer to the assumptions stated in \citet[Section 2]{KinFriLil15} and \citet[Theorem 4.3.8]{DehFer06}. \\
The following statistics are computed from observations of scheme \eqref{eq:sample}, where $a_j=\lfloor n(1-r_j)\rfloor$ such that $n_j/n\rightarrow r_j>0$ for $n\rightarrow\infty$ and with random vectors $\mathbf{X}_i=(X_{i,1},\ldots,X_{i,d})'$ having marginal distribution functions $F_j=\mathbb{P}(X_{i,j}\leq x)$ of Pareto-type. The random vector $\mathbf{H}_{\mathbf{k},\bm{r},n}$ is a collection of local Hill-estimators, with $j$-th component denoted by $H_{k_j,r_j,n}$ computed from the $k_j$ largest order statistics of the sample $X_{a_j+1,j},X_{a_j+2,j},\ldots,X_{n,j}$.

\begin{proposition}\label{prop:regionalEVI}\textsl{(Proposition 1 in \citet{KinFriLil15})} \\
Suppose that $\gamma_1=\ldots=\gamma_d=\gamma$, $k_1\rightarrow\infty$, $k_1/n\rightarrow0$ and $k_1/k_\ell\rightarrow c_\ell$ hold for $n\rightarrow\infty$ and finite values $c_\ell>0$, $\ell=2,\ldots,d$. Let $\mathbf{1}=(1,\ldots,1)'\in\mathbb{R}^d$. Then, under additional technical assumptions \citep[Sec. 2]{KinFriLil15} and for $n\rightarrow\infty$, we have that
\begin{align}
\sqrt{k_1}\left(\mathbf{H}_{\mathbf{k},\bm{r},n}-\gamma\mathbf{1}\right)\stackrel{D}{\longrightarrow}\mathcal{N}\left(0,\gamma^2\cdot\Sigma\right) \nonumber
\end{align}
holds, where $\Sigma\in\mathbb{R}^{d\times d}$ is defined componentwise by
\begin{align}
\Sigma_{l,m} = c_l\cdot c_m\cdot(r_l\wedge r_m)\cdot\Lambda_{l,m}\left((r_lc_l)^{-1},(r_mc_m)^{-1}\right) \nonumber
\end{align}
for $1\leq l, m\leq d$ and where $x\wedge y=\min(x,y)$. For $l=m$ this reduces to $\Sigma_{l,l}=c_l$.
\end{proposition}

For the estimation of $\Sigma$ we plug in $\hat{c}_{\ell}=\frac{k_1}{k_\ell}$ and a consistent estimator $\hat{\Lambda}_{\ell,m}$ of the upper tail dependence copula $\Lambda_{\ell,m}$ between components $\ell$ and $m$. For the latter, we apply two different solutions: The empirical estimator studied in \citet{schmidt_non-parametric_2006} or, under the extreme value dependence assumption, we use representation
\begin{align}
\Lambda_{\ell,m}(x,y)=(x+y)\cdot\left[1-A_{\ell,m}\left(\frac{y}{x+y}\right)\right] \nonumber
\end{align}
with corresponding Pickands function $A_{\ell,m}$ and plug in the corrected CFG-estimator $\hat{A}_{\ell,m}$ studied in \citet{genest_rank-based_2009}. As a direct consequence of Proposition \ref{prop:regionalEVI}, we have
\begin{align}
\sqrt{k_1}\left(\hat{\gamma}_{\mathbf{k},\bm{r},n}(\mathbf{w})-\gamma\right)\stackrel{D}{\rightarrow}\mathcal{N}(0,\gamma^2\mathbf{w}'\Sigma\mathbf{w}), \nonumber
\end{align}
where $\hat{\gamma}_{\mathbf{k},\bm{r},n}(\mathbf{w})=\mathbf{w}'\mathbf{H}_{\mathbf{k},\bm{r},n}$ is called regional estimator of $\gamma$.

\medskip

Let $\hat{F}_j^{-1}(p)$ denote the regional estimator of $F_j^{-1}(p)$ defined in \eqref{eq:jointquantile}.

\begin{proposition}\label{prop:weissmanconverg}\textsl{(we refer to \citet[Sec. 2.7]{Kin16})} \\
Let $1\leq j\leq d$ be fixed. Let $p=p_n\in(0,1)$ be a sequence of numbers such that $p\rightarrow1$, $n(1-p)/k_1\rightarrow0$ and $\log(np)/\sqrt{k_1}\rightarrow0$ hold for $n\rightarrow\infty$. Then, under the assumptions from Proposition \ref{prop:regionalEVI}, the same technical assumptions as in Proposition~\ref{prop:regionalEVI} and for $n\rightarrow\infty$, we have that
\begin{align}
\frac{\sqrt{k_j}}{\log\frac{k_j}{n_j(1-p)}}\left(\frac{\hat{F}_j^{-1}(p)}{F_j^{-1}(p)}-1\right)\stackrel{D}{\rightarrow}\mathcal{N}\left(0,\frac{\gamma^2}{c_j}\mathbf{w}'\Sigma\mathbf{w}\right), \nonumber
\end{align}
where $\hat{F}_j^{-1}(p)$ is defined in \eqref{eq:weissmanextrapolation}.
\end{proposition}
Proposition \ref{prop:weissmanconverg} allows us to derive an asymptotic $(1-\alpha)$-confidence interval 
\begin{align}\label{eq:ci}
CI^{(1)}(\alpha)=\hat{F}_j^{-1}(p)\cdot\left[1\pm z_{1-\alpha/2}\cdot\sqrt{\frac{\hat{\gamma}_{\mathbf{k},\bm{r},n}^2}{k_1}\mathbf{w}'\hat{\Sigma}\mathbf{w}}\cdot\log\left(\frac{k_j}{n_j(1-p)}\right)\right],
\end{align}
where $z_{1-\alpha/2}$ is the $(1-\alpha/2)$-quantile of the standard normal distribution and with $a\cdot(1\pm b)$ denoting the interval $[a\cdot(1-b),\ a\cdot(1+b)]$ for $a\in\mathbb{R}$, $b>0$. 

\bibliographystyle{chicago}
\bibliography{biblio}

\begin{thebibliography}{}

\bibitem[\protect\citeauthoryear{Beirlant, Goegebeur, Segers, and
  Teugels}{Beirlant et~al.}{2006}]{beirlant_statistics_2006}
Beirlant, J., Y.~Goegebeur, J.~Segers, and J.~Teugels (2006).
\newblock {\em Statistics of Extremes: Theory and Applications}.
\newblock John Wiley \& Sons.

\bibitem[\protect\citeauthoryear{{B{\"u}cher} and {Segers}}{{B{\"u}cher} and
  {Segers}}{2015}]{BueSeg15}
{B{\"u}cher}, A. and J.~{Segers} (2015, November).
\newblock {Maximum likelihood estimation for the Fr$\backslash$'echet
  distribution based on block maxima extracted from a time series}.
\newblock {\em ArXiv:1511.07613\/}.

\bibitem[\protect\citeauthoryear{Cunderlik and Burn}{Cunderlik and
  Burn}{2006}]{CunBur06}
Cunderlik, J.~M. and D.~H. Burn (2006).
\newblock Switching the pooling similarity distances: Mahalanobis for
  euclidean.
\newblock {\em Water Resources Research\/}~{\em 42\/}(3).
\newblock W03409.

\bibitem[\protect\citeauthoryear{Dalrymple}{Dalrymple}{1960}]{Dal60}
Dalrymple, T. (1960).
\newblock Flood-frequency analyses, manual of hydrology: Part 3.
\newblock Technical report, USGPO,.

\bibitem[\protect\citeauthoryear{de~Haan and Ferreira}{de~Haan and
  Ferreira}{2006}]{DehFer06}
de~Haan, L. and A.~Ferreira (2006).
\newblock {\em Extreme Value Theory: An Introduction\/} (Auflage: 2006 ed.).
\newblock Springer.

\bibitem[\protect\citeauthoryear{Dematteo and {Cl{\'e}men{\c c}on}}{Dematteo
  and {Cl{\'e}men{\c c}on}}{2015}]{DemCle15}
Dematteo, A. and S.~{Cl{\'e}men{\c c}on} (2015).
\newblock On tail index estimation based on multivariate data.
\newblock {\em Journal of Nonparametric Statistics\/}~{\em 0\/}(0), 1--25.

\bibitem[\protect\citeauthoryear{{DWA}}{{DWA}}{2012}]{dwa_merkblatt_2012}
{DWA} (2012).
\newblock {\em Merkblatt {DWA}-M 552, Ermittlung von
  Hochwasserwahrscheinlichkeiten}.
\newblock Deutsche Vereinigung fuer Wasserwirtschaft, Abwasser und Abfall e.V.

\bibitem[\protect\citeauthoryear{Elamir and Seheult}{Elamir and
  Seheult}{2003}]{ElaSeh03}
Elamir, E.~A. and A.~H. Seheult (2003).
\newblock Trimmed l-moments.
\newblock {\em Computational Statistics and Data Analysis\/}~{\em 43\/}(3), 299
  -- 314.

\bibitem[\protect\citeauthoryear{Ferreira and de~Haan}{Ferreira and
  de~Haan}{2015}]{FerDeh15}
Ferreira, A. and L.~de~Haan (2015, 02).
\newblock On the block maxima method in extreme value theory: Pwm estimators.
\newblock {\em Ann. Statist.\/}~{\em 43\/}(1), 276--298.

\bibitem[\protect\citeauthoryear{Fisher and Tippett}{Fisher and
  Tippett}{1928}]{FisTip1928}
Fisher, R.~A. and L.~H.~C. Tippett (1928).
\newblock Limiting forms of the frequency distribution of the largest or
  smallest member of a sample.
\newblock {\em Mathematical Proceedings of the Cambridge Philosophical
  Society\/}~{\em 24}, 180--190.

\bibitem[\protect\citeauthoryear{Gabriele and Arnell}{Gabriele and
  Arnell}{1991}]{GabArn91}
Gabriele, S. and N.~Arnell (1991).
\newblock A hierarchical approach to regional flood frequency analysis.
\newblock {\em Water Resources Research\/}~{\em 27\/}(6), 1281--1289.

\bibitem[\protect\citeauthoryear{Genest and Segers}{Genest and
  Segers}{2009}]{genest_rank-based_2009}
Genest, C. and J.~Segers (2009).
\newblock Rank-based inference for bivariate extreme-value copulas.
\newblock {\em Ann. Statist.\/}~{\em 37\/}(5B), 2990--3022.

\bibitem[\protect\citeauthoryear{Hill}{Hill}{1975}]{Hil75}
Hill, B.~M. (1975).
\newblock A simple general approach to inference about the tail of a
  distribution.
\newblock {\em Ann. Statist.\/}~{\em 3\/}(5), 1163--1174.

\bibitem[\protect\citeauthoryear{Hosking}{Hosking}{2007}]{Hos07}
Hosking, J. (2007).
\newblock Some theory and practical uses of trimmed l-moments.
\newblock {\em Journal of Statistical Planning and Inference\/}~{\em 137\/}(9),
  3024 -- 3039.

\bibitem[\protect\citeauthoryear{Hosking}{Hosking}{1990}]{Hos90}
Hosking, J. R.~M. (1990).
\newblock L-moments: Analysis and estimation of distributions using linear
  combinations of order statistics.
\newblock {\em Journal of the Royal Statistical Society. Series B
  (Methodological)\/}~{\em 52\/}(1), 105--124.

\bibitem[\protect\citeauthoryear{Hosking, Wallis, and Wood}{Hosking
  et~al.}{1985}]{HosWalWoo85}
Hosking, J. R.~M., J.~R. Wallis, and E.~F. Wood (1985).
\newblock Estimation of the generalized extreme-value distribution by the
  method of probability-weighted moments.
\newblock {\em Technometrics\/}~{\em 27\/}(3), pp. 251--261.

\bibitem[\protect\citeauthoryear{Khoudraji}{Khoudraji}{1995}]{Kho95}
Khoudraji, A. (1995).
\newblock {\em Contributions \`a l'\'etude des copules et \`a la mod\'elisation
  des valeurs extr\^emes bivari\'ees}.
\newblock Ph.\ D. thesis, Universit\'e Laval, Qu\'ebec, Canada.

\bibitem[\protect\citeauthoryear{Kinsvater}{Kinsvater}{2016}]{Kin16}
Kinsvater, P. (2016).
\newblock {\em Semi- and non-parametric flood frequency analysis}.
\newblock Ph.\ D. thesis, TU Dortmund, Dortmund, Germany.

\bibitem[\protect\citeauthoryear{Kinsvater, Fried, and Lilienthal}{Kinsvater
  et~al.}{2016}]{KinFriLil15}
Kinsvater, P., R.~Fried, and J.~Lilienthal (2016).
\newblock Regional extreme value index estimation and a test of tail
  homogeneity.
\newblock {\em Environmetrics\/}~{\em 27\/}(2), 103--115.

\bibitem[\protect\citeauthoryear{Kosorok}{Kosorok}{2008}]{Kos08}
Kosorok, M.~R. (2008).
\newblock {\em Introduction to empirical processes and semiparametric
  inference}.
\newblock Springer Series in Statistics. New York: Springer.

\bibitem[\protect\citeauthoryear{Lettenmaier, Wallis, and Wood}{Lettenmaier
  et~al.}{1987}]{LetWalWoo87}
Lettenmaier, D.~P., J.~R. Wallis, and E.~F. Wood (1987).
\newblock Effect of regional heterogeneity on flood frequency estimation.
\newblock {\em Water Resources Research\/}~{\em 23\/}(2), 313--323.

\bibitem[\protect\citeauthoryear{Lilienthal, Kinsvater, and Fried}{Lilienthal
  et~al.}{2016}]{LilKinFri16}
Lilienthal, J., P.~Kinsvater, and R.~Fried (2016).
\newblock On the method of probability weighted moments in regional frequency
  analysis.
\newblock {\em SFB 823 Discussion\/}~{\em 63/2016}.

\bibitem[\protect\citeauthoryear{Martins and Stedinger}{Martins and
  Stedinger}{2002}]{MarSte02}
Martins, E.~S. and J.~R. Stedinger (2002).
\newblock Cross correlations among estimators of shape.
\newblock {\em Water Resources Research\/}~{\em 38\/}(11), 34--1--34--7.
\newblock 1252.

\bibitem[\protect\citeauthoryear{Rossi, Fiorentino, and Versace}{Rossi
  et~al.}{1984}]{RosFioVer84}
Rossi, F., M.~Fiorentino, and P.~Versace (1984).
\newblock Two-component extreme value distribution for flood frequency
  analysis.
\newblock {\em Water Resources Research\/}~{\em 20\/}(7), 847--856.

\bibitem[\protect\citeauthoryear{Rulfov{\'{a}}, Buishand, Roth, and
  Kysel{\'{y}}}{Rulfov{\'{a}} et~al.}{2016}]{RulBuiRotKys16}
Rulfov{\'{a}}, Z., A.~Buishand, M.~Roth, and J.~Kysel{\'{y}} (2016).
\newblock A two-component generalized extreme value distribution for
  precipitation frequency analysis.
\newblock {\em Journal of Hydrology\/}~{\em 534}, 659 -- 668.

\bibitem[\protect\citeauthoryear{Schmidt and Stadtm{\"u}ller}{Schmidt and
  Stadtm{\"u}ller}{2006}]{schmidt_non-parametric_2006}
Schmidt, R. and U.~Stadtm{\"u}ller (2006).
\newblock Non-parametric estimation of tail dependence.
\newblock {\em Scandinavian Journal of Statistics\/}~{\em 33\/}(2), pp.
  307--335.

\bibitem[\protect\citeauthoryear{Stedinger}{Stedinger}{1983}]{Ste83}
Stedinger, J.~R. (1983).
\newblock Estimating a regional flood frequency distribution.
\newblock {\em Water Resources Research\/}~{\em 19\/}(2), 503--510.

\bibitem[\protect\citeauthoryear{Strupczewski, Kochanek, Bogdanowicz, and
  Markiewicz}{Strupczewski et~al.}{2012}]{StrKoc12}
Strupczewski, W.~G., K.~Kochanek, E.~Bogdanowicz, and I.~Markiewicz (2012).
\newblock On seasonal approach to flood frequency modelling. part i:
  Two-component distribution revisited.
\newblock {\em Hydrological Processes\/}~{\em 26\/}(5), 705--716.

\bibitem[\protect\citeauthoryear{Van~der Vaart}{Van~der Vaart}{2000}]{Van2000}
Van~der Vaart, A.~W. (2000).
\newblock {\em Asymptotic statistics}, Volume~3.
\newblock Cambridge university press.

\bibitem[\protect\citeauthoryear{Weissman}{Weissman}{1978}]{Wei78}
Weissman, I. (1978).
\newblock Estimation of parameters and large quantiles based on the k largest
  observations.
\newblock {\em Journal of the American Statistical Association\/}~{\em
  73\/}(364), 812--815.

\end{thebibliography}

\end{document}